\documentclass[12pt]{article}
\usepackage{amsmath}
\usepackage{graphicx}
\usepackage{enumerate}
\usepackage{natbib}

\newcommand{\blind}{0}

\usepackage{soul, xcolor} 
\usepackage{amsfonts}
\usepackage{amsthm}
\usepackage{bm}
\usepackage[labelformat=simple]{subcaption}
\usepackage{hyperref}

\newcommand{\R}{\mathbb{R}}
\newcommand{\N}{\mathbb{N}}
\newcommand{\E}{\mathbb{E}}
\newcommand{\p}{\mathbb{P}}
\newcommand{\norm}[1]{\lVert#1\rVert}
\newcommand{\abs}[1]{\lvert#1\rvert}
\newcommand{\betahat}{\hat{\beta}}
\newcommand{\betahatlone}{\hat{\beta}^{\ell_1}}
\newcommand{\Shat}{\hat{S}}
\newcommand{\gammahat}{\hat{\gamma}}
\newcommand{\comp}[1]{\overline{#1}}
\DeclareMathOperator*{\argmin}{arg\,min}
\DeclareMathOperator*{\sign}{sign}

\DeclareMathOperator*{\rank}{rank}

\DeclareMathOperator*{\diag}{diag}

\newtheoremstyle{thm}
{}
{}
{\itshape}
{}
{\bf}
{.}
{ }
{}

\newtheoremstyle{text}
{}
{}
{\upshape}
{}
{\bf}
{.}
{ }
{}

\theoremstyle{thm}
\newtheorem{lemma}{Lemma}
\newtheorem{prop}{Proposition}
\newtheorem{theorem}{Theorem}
\newtheorem{corollary}{Corollary}

\theoremstyle{text}

\newtheorem{defn}{Definition}

\begin{document}

\def\spacingset#1{\renewcommand{\baselinestretch}%
{#1}\small\normalsize} \spacingset{1}


\if0\blind
{
  \title{\bf Model selection with Lasso-Zero: adding straw to the haystack to better find needles}
  \author{Pascaline Descloux\thanks{pascaline.descloux@unige.ch}\hspace{.2cm}\\
    and \\
    Sylvain Sardy\thanks{sylvain.sardy@unige.ch} \\
    Department of Mathematics, University of Geneva, Switzerland}
    \date{}
  \maketitle
} \fi

\if1\blind
{
  \bigskip
  \bigskip
  \bigskip
  \begin{center}
    {\LARGE\bf Title}
\end{center}
  \medskip
} \fi

\bigskip
\begin{abstract}
The high-dimensional linear model $y = X \beta^0 + \epsilon$ is considered and the focus is put on the problem of recovering the support $S^0$ of the sparse vector~$\beta^0.$ We introduce \emph{Lasso-Zero}, a new $\ell_1$-based estimator whose novelty resides in an ``overfit, then threshold" paradigm and the use of noise dictionaries concatenated to $X$ for overfitting the response. To select the threshold, we employ the quantile universal threshold based on a pivotal statistic that requires neither knowledge nor preliminary estimation of the noise level. Numerical simulations show that Lasso-Zero performs well in terms of support recovery and provides an excellent trade-off between high true positive rate and low false discovery rate compared to competitors. Our methodology is supported by theoretical results showing that when no noise dictionary is used, Lasso-Zero recovers the signs of $\beta^0$ under weaker conditions on $X$ and $S^0$ than the Lasso and achieves sign consistency for correlated Gaussian designs. The use of noise dictionary  improves the procedure for low signals.
\end{abstract}

\noindent%
\vfill

\newpage
\spacingset{1.5} 

\section{Introduction}
\subsection{Recovery of a sparse vector in the linear model}
Nowadays in many statistical applications the number $p$ of parameters exceeds the number $n$ of observations. In such a high-dimensional setting, additional assumptions  are necessary for estimation of the parameters to be possible. For the linear model
\begin{equation} \label{lin_model}
y = X \beta^0 + \epsilon,
\end{equation}
where $X$ is a matrix of size $n \times p$ and $\epsilon \sim N_n(0, \sigma^2 I)$ is the noise component, the coefficient vector $\beta^0 \in \R^p$ is often assumed to be sparse, meaning that most of its entries are zero. In a regression setting where each column $X_j \in \R^n$ of $X$ corresponds to a potential predictor and $y$ is the response variable, sparsity means that only a few of the predictors at hand are relevant for predicting $y$. Genetics provides a typical modern example, with datasets containing expression levels of thousands of genes for only a few observational units, and where researchers are often interested in identifying genes related to a certain medical trait. Another example is the problem of recovering discontinuity points in a piecewise constant function $f$ given a noisy observation, which can also be formulated in the form~(\ref{lin_model}), where a nonzero component in $\beta^0$ corresponds to a discontinuity in $f.$ These examples have the common goal to correctly identify the support $S^0:= \{j \mid \beta^0_j \neq 0 \}$ of $\beta^0$. 
A vast amount of literature is dedicated to this problem, in the noisy as well as in the noiseless ($\sigma = 0$) case. In the noiseless case where $y = X \beta^0$, it is ideally desired to recover the sparsest vector $\beta$ satisfying $y = X \beta$, i.e. the solution to
\begin{equation} \label{l0_problem}
\begin{aligned}
& \min_{\beta \in \R^p} && \norm{\beta}_0 \\
& \text{ s.t. } && y = X \beta,
\end{aligned}
\end{equation}
where $\norm{\beta}_0$ denotes the number of nonzero coefficients of $\beta.$ The problem~(\ref{l0_problem}) being NP-hard \citep{natarajan1995}, the convex relaxation
\begin{equation} \label{basis_pursuit}
\begin{aligned}
& \min_{\beta \in \R^p} && \norm{\beta}_1 \\
 &  \text{ s.t.} &&  y = X \beta,
\end{aligned}
\end{equation}
called the basis pursuit (BP) problem \citep{chen2001}, provides a computationally attractive alternative. Problem~(\ref{basis_pursuit}) used as a proxy for~(\ref{l0_problem}) as been well studied in mathematical signal processing (see \citet{foucart2013} and references therein) and it is known that under some conditions the solutions of~(\ref{basis_pursuit}) are solutions to~(\ref{l0_problem}) as well \citep{donoho2001, elad2002, donoho2003, gribonval2003, fuchs2004, donoho2006, candes2006robust}.

When the observations are corrupted by noise, the most natural adaptation of~(\ref{basis_pursuit}) is to modify the constraint to allow $\norm{y - X \beta}_2$ to be positive, but not greater than a certain value. This yields basis pursuit denoising \citep{candes2006stable, donoho2006stable}:
\begin{equation} \label{basis_pursuit_denoising}
\begin{aligned}
& \min_{\beta \in \R^p} && \norm{\beta}_1 \\
& \text{ s.t.} &&  \norm{y - X \beta}_2 \leq \theta.
\end{aligned}
\end{equation}
Basis pursuit denoising is known (see for example \citet{foucart2013}) to be equivalent to the Lasso estimator \citep{tibshirani1996}
\begin{equation} \label{lasso}
\begin{aligned}
&\betahat^{\rm{lasso}}_\lambda = \argmin_{\beta \in \R^p} \frac{1}{2} \norm{y - X \beta}_2^2 + \lambda \norm{\beta}_1, \quad \lambda > 0.
\end{aligned}
\end{equation}
The Lasso has been extensively studied in the literature (see \citet{buhlmann2011} and references therein). Under some conditions, it is known to detect all important variables with high probability for an appropriate choice of $\lambda$. However it often includes too many false positives \citep{su2017}, and exact support recovery can be achieved only under a strong condition on $X$ and $\sign(\beta^0)$, called the irrepresentable condition \citep{zhao2006, zou2006, buhlmann2011}. 

Many regularized estimators based on the $\ell_1$-penalty have been proposed. Some like adaptive Lasso \citep{zou2006}, thresholded Lasso \citep{vandegeer2011} and relaxed Lasso \citep{meinshausen2007} add a second selection step after solving the Lasso. The group Lasso \citep{yuan2005} selects grouped variables; the elastic net \citep{zou2005} combines the $\ell_1$ and $\ell_2$ types of penalty; the Dantzig selector \citep{candes2007} minimizes the $\ell_1$-norm subject to a constraint on the maximal correlation between $X$ and the residual vector. Resampling procedures like stability selection \citep{meinshausen2010} and more recent techniques like the knockoff filter \citep{barber2015, candes2018} and SLOPE \citep{bogdan2015} have been introduced to control some measures of the type I error. Although differing in their purposes and performance, the general idea underlying these procedures remains the same, namely to avoid overfitting by finding a trade-off between the fit $y - X \beta$ and some measure of the model complexity.

\subsection{Our approach and contribution}
Relying on the success of BP~(\ref{basis_pursuit}) in the noiseless case, we argue that there is another way to adapt to the presence of noise. A naive alternative to~(\ref{basis_pursuit_denoising}) is to overfit the response $y$ by solving BP in a first step, and then to threshold the obtained solution to retain only the largest coefficients. We prove in Section~\ref{theory} that this procedure requires less stringent conditions on $X$ and $S^0$ than the Lasso for exact support recovery. 

Figure~\ref{lasso_path} illustrates a typical example where $S^0$ can be recovered exactly by a thresholded BP solution, but not by the Lasso. It corresponds to the simulation setting~(\ref{wideGauss}) described in Section~\ref{settings} with a support $S^0$ of size $\abs{S^0} = 10$. The curves represent all components $\betahat^{\rm{lasso}}_{\lambda, j}$ of the Lasso solution as $\lambda$ varies. The red curves correspond to indices $j$ belonging to the support $S^0$. Looking at the path it is clear that the Lasso cannot recover the true support in this case since there is no value of $\lambda$ for which the estimated support equals~$S^0$. On the other hand, one does see a clear separation between the red and black curves if one looks at the limit solution of the path as $\lambda$ tends to zero. Interestingly, the limit $\lim_{\lambda \to 0^+} \betahat^{\rm{lasso}}_\lambda$ of the Lasso path is known to be a solution to BP \citep{fuchs2004}. So in this case there exists a threshold level $\tau$ -- e.g. the one indicated horizontally in Figure~\ref{lasso_path} -- such that the thresholded BP solution is exactly supported on $S^0.$
\begin{figure}
\begin{center}
\includegraphics[scale=0.3]{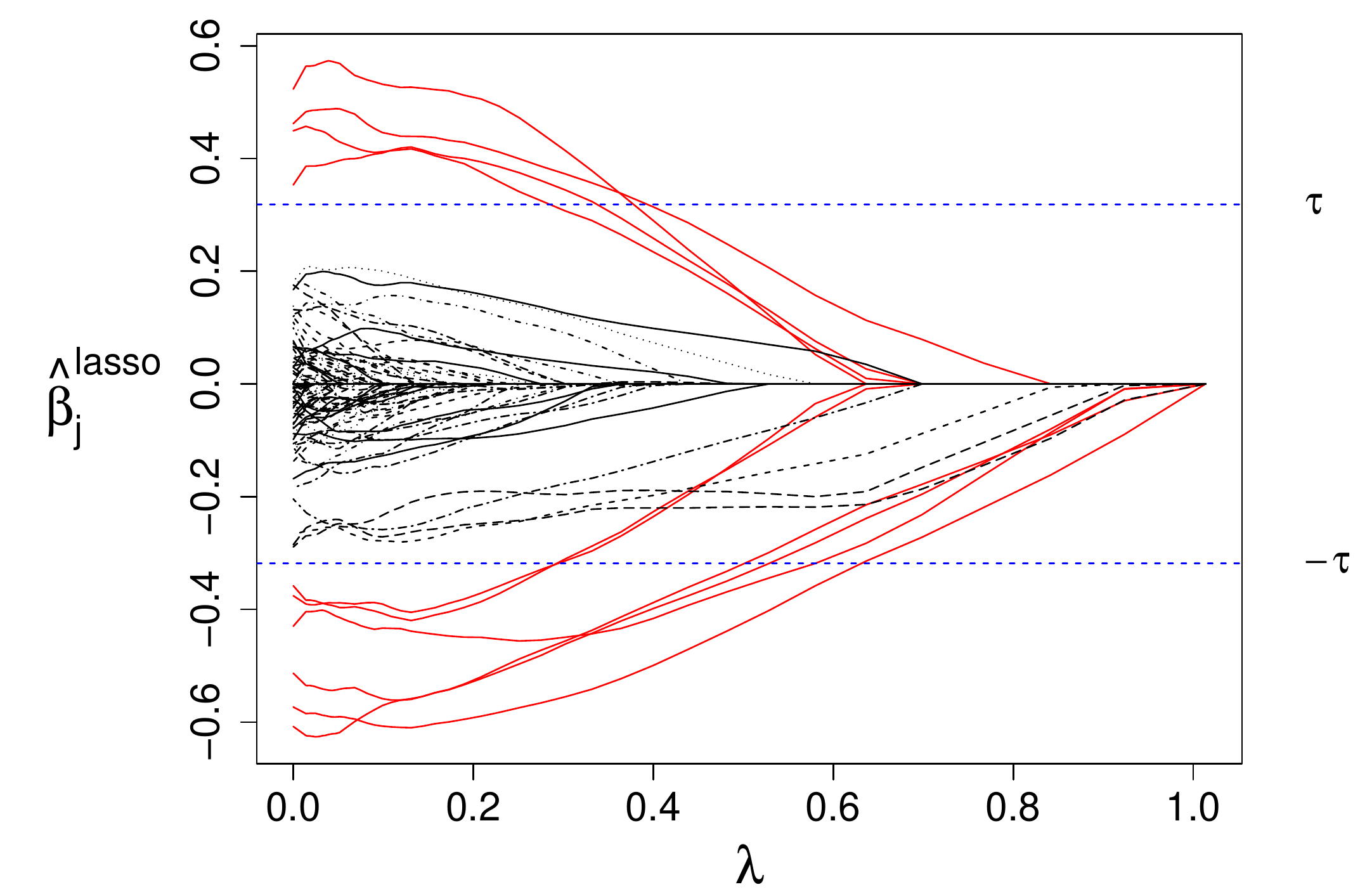}
\caption{Path of all Lasso solutions for a signal generated according to setting~(\ref{wideGauss}) in Section~\ref{settings}. The support $S^0$ has size $10$ and is represented by the red curves. None of the Lasso solutions recover $S^0$, whereas the BP solution -- corresponding to the limit Lasso solution as $\lambda$ tends to $0$ -- recovers $S^0$ exactly if thresholded properly.}
\label{lasso_path}
\end{center}
\end{figure} 

Motivated by this observation, we keep in mind an ``overfit, then threshold" paradigm and introduce Lasso-Zero, an $\ell_1$-based estimator improving the procedure described above in cases where the signal is too low for basis pursuit to extract the right coefficients. Apart from focusing on the Lasso solution at $\lambda = 0$ -- which motivated its name -- the novelty of Lasso-Zero also resides in the use of several random noise dictionaries in the overfitting step followed by the aggregation of the corresponding estimates.

The choice of the threshold level is crucial. A good threshold achieves a high true positive rate (TPR) for reasonable signal-to-noise ratio while maintaining a low false discovery rate (FDR). Also, a good threshold selection can be used even if the noise level~$\sigma$ is unknown: we opt for quantile universal thresholding~\citep{giacobino2017}, a procedure that controls FDR under the null model~$\beta^0 = 0$. 
Our numerical simulations show that Lasso-Zero  provides an excellent trade-off between low FDR and high TPR and in many cases outperforms its competitors in terms of exact support recovery.

\subsection{Related work}
The methodology proposed in this paper is an extension of the thresholded basis pursuit estimator, which has been little studied in the literature. \citet{saligrama2011} prove that the multistage procedure consisting in 1) solving BP, 2) thresholding, 3) computing the least-squares solution on the obtained support and 4) thresholding again, consistently recovers the signs of $\beta^0$ for i.i.d.~Gaussian designs. Our methodology differs from theirs in that we suggest the repeated use of noise dictionaries to improve thresholded BP rather than following their steps 3) and 4). Secondly, their results assume some restricted isometry property which is stronger than our assumption and allows no direct comparison to the necessary conditions required by the Lasso for exact model selection. Thirdly, we propose a way to select the threshold level in practice, unlike the aforementioned work in which the threshold depends on the (unknown) minimal nonzero coefficient in $\beta^0.$ More recently, \citet{tardivel2018} analyzed thresholded BP under a weak identifiability assumption and obtained results similar to the ones we present in Section~\ref{subsection:analysis_SNSP}.

\subsection{Organization of the paper}
The Lasso-Zero estimator is defined in Section~\ref{lasso-zero}. The choice of the threshold $\tau$ is then discussed in Section~\ref{thresh_selection}: the quantile universal threshold (QUT) for Lasso-Zero is derived in Section~\ref{qut}, and Section~\ref{GEV} suggests to rely on a generalized extreme value (GEV) parametric fit to approximate QUT. Results of numerical simulations are presented in Section~\ref{simus}. Section~\ref{theory} presents theoretical results motivating our methodology, namely that thresholding the BP solution recovers $\sign(\beta^0)$ under weaker conditions on the design and the support than the Lasso (Section~\ref{subsection:analysis_SNSP}), and consequently achieves sign consistency for correlated Gaussian designs (Section~\ref{subsection:sign_consistency_Gauss}). We prove guarantees offered by QUT in terms of $\operatorname{FWER}$ and $\operatorname{FDR}$ control in Section~\ref{lowdim_FDR}. A final discussion is given in Section~\ref{discussion}.

\vspace{1em}
\textbf{Notation:} For a matrix $X \in \R^{n \times p}$ and a subset $S \subset \{1, \ldots, p\}$ of size $s$, the submatrix of size $n \times s$ consisting of all columns of $X$ indexed by elements $j \in S$ is denoted $X_S$. Similarly, for a vector $\beta \in \R^p$, $\beta_S$ denotes the vector in $\R^s$ obtained by extracting all components of $\beta$ indexed by $S$. For two vectors $\beta, \tilde{\beta} \in \R^p,$ $\beta \overset{\operatorname{s}}{=} \tilde{\beta}$ indicates that the signs of $\beta$ and $\tilde{\beta}$ are componentwise equal.

\section{Lasso-Zero} \label{lasso-zero}
To recover the support $S^0$ of $\beta^0$ in the noisy linear model~(\ref{lin_model}), we could ignore the noise and solve the BP problem~(\ref{basis_pursuit}) in a first step, and in a second step threshold the obtained solution so that only the large coefficients remain. It turns out that if the perturbation $\epsilon$ is small enough, this procedure can recover $S^0$ exactly under a weak assumption on $X$ and $S^0$ (see Section~\ref{theory}). However in practice the signal is often too low and it is not desirable to enforce the exact fit $y=X \beta$ as it is constrained by BP. We therefore suggest to use a \emph{noise dictionary} consisting in a random matrix $G \in \R^{n \times q}$ to fit the noise term. Indeed if $\rank{G} = n$ then there exists a vector $\gamma_{\epsilon, G} \in \R^q$ such that $\epsilon = G \gamma_{\epsilon, G}$ and the equality in model~(\ref{lin_model}) can be rewritten as
$
y = X\beta^0 + G\gamma_{\epsilon, G}.
$
One obtains estimates for $\beta^0$ and $\gamma_{\epsilon, G}$ by solving BP for the extended matrix $(X \mid G),$ i.e.
\begin{equation} \label{strawlass0}
\begin{aligned}
 (\betahat, \gammahat) = & && \argmin_{\beta \in \R^p, \gamma \in \R^q} && \norm{\beta}_1 + \norm{\gamma}_1 \\
& &&\text{\quad s.t. } && y = X\beta + G \gamma.
\end{aligned}
\end{equation}
If the solution $(\betahat, \gammahat)$ is unique then it contains at most $n$ nonzero entries \citep[Theorem 3.1]{foucart2013}. So the extended BP problem (\ref{strawlass0}) selects some columns of $X$ and $G$ to fit the response $y,$ which is the sum of signal and noise. It now becomes clear that the purpose of the random dictionary $G$ is to provide new columns that can be selected to fit the noise term $\epsilon,$ so that the columns of $X$ can be mostly used to fit the signal $X \beta^0$. For this selection to be fair the columns of $X$ and $G$ should be standardized to be on the same scale, e.g. to have the same $\ell_2$-norm (see also our remark at the end of Section~\ref{settings}). 

There is no guarantee that $\betahat$ is supported on $S^0$: for instance, if by misfortune some columns of $G$ are strongly correlated with a column $X_j$ with $j \in S^0$, then BP might set~$\betahat_j$ to zero in favor of columns of $G$. A way to circumvent this is to repeat the procedure several times and take the median for each component of the estimated $\beta$s, since for well behaved $X$ it can be expected that most of the times $\betahat_j$ will be set to zero for $j \in \comp{S^0}$ and nonzero for $j \in S^0$. This gives the following estimator.

\begin{defn} \label{deflass0}
For given $q \in \N, M \in \N^\star$ and for a given threshold $\tau \geq 0$, the \emph{Lasso-Zero} estimator $\betahat^{\operatorname{lass0}(q, M)}_\tau$ is defined as follows:
\begin{enumerate}
\item For $k= 1, \ldots, M:$ generate $G^{(k)} \in \R^{n \times q}$ with entries $G^{(k)}_{ij} \overset{\textrm{i.i.d.}}{\sim} N(0, 1),$ then compute the solution $(\betahat^{(k)}, \gammahat^{(k)})$ to~(\ref{strawlass0}) with $G = G^{(k)}.$
\item \label{lass0_before_thresh} Define 
$
\betahat^{\ell_1}_j := \operatorname{median} 
\{ \betahat^{(k)}_j, k=1, \ldots, M \}$   for $j \in \{1, \ldots, p\}$.  
\item Threshold the coefficients at level $\tau$ to obtain
$\betahat^{\operatorname{lass0}(q,M)}_{\tau, j} := \eta_\tau (\betahat^{\ell_1}_j) $ for $j \in \{1, \ldots, p\}$,
where $\eta_\tau$ denotes any thresholding function satisfying $\eta_\tau(x) = 0$ if $\abs{x} \leq \tau$ and $\sign(\eta_\tau(x)) = \sign(x)$ otherwise; typically soft or hard thresholding are used.
\end{enumerate}
We will often omit the parameters $(q,M)$ and write $\betahat^{\operatorname{lass0}}_\tau$ to avoid heavy notation. The obtained estimated support is denoted
$
\Shat^{\operatorname{lass0}}_\tau := \{j \mid \betahat^{\operatorname{lass0}}_{\tau, j} \neq 0 \}.
$
\end{defn}


For the sake of generality Lasso-Zero is defined for any $q \geq 0.$ This allows us to consider Lasso-Zero as an extension of thresholded basis pursuit \citep{saligrama2011}, which is obtained with $q=0$ and $M=1$ (see Section~\ref{theory}). In practice, since the dictionaries $G^{(k)} \in \R^{n \times q}$ aim at fitting the noise $\epsilon,$ it is desirable to choose $q \geq n.$ In this case problem~(\ref{strawlass0}) is almost surely feasible for any matrix $X.$ We suggest and use $q=n$ in Section~\ref{simus}. 

\renewcommand{\thesubfigure}{(\roman{subfigure})}
\begin{figure}[!ht]
	\centering
	\begin{subfigure}[t]{7cm}
		\centering
		\includegraphics[scale=0.27]{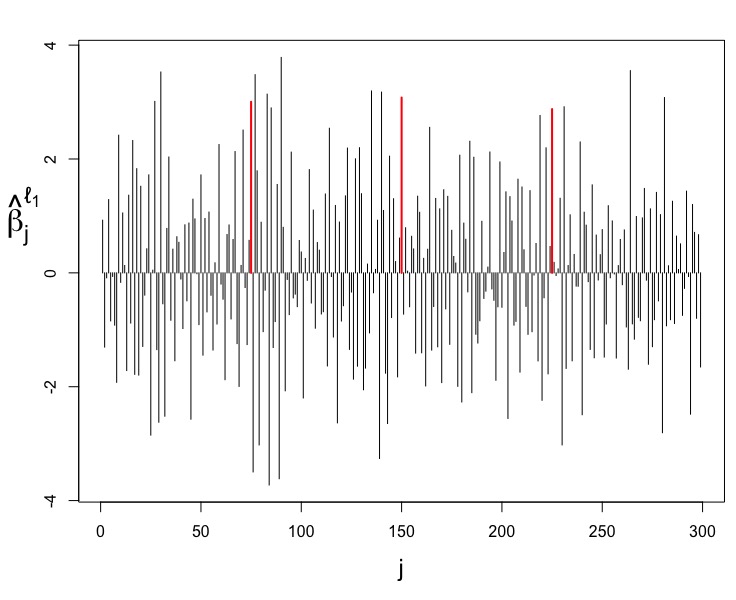}
		\caption{$(q, M)=(0, 1)$\label{seg_BP}}		
	\end{subfigure}
	\qquad
	\begin{subfigure}[t]{7cm}
		\centering
		\includegraphics[scale=0.27]{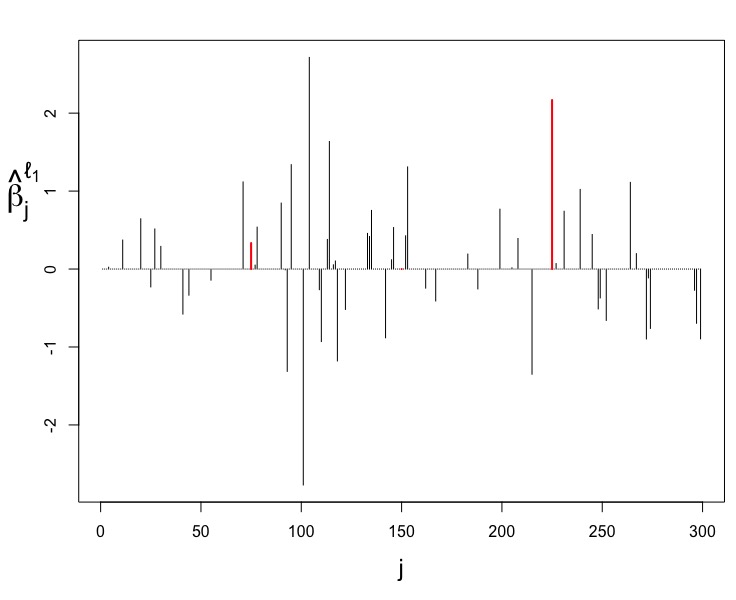}
		\caption{$(q, M) = (n,1)$\label{seg_M1}}
	\end{subfigure}
\quad
	\begin{subfigure}[t]{7cm}
		\centering
		\includegraphics[scale=0.27]{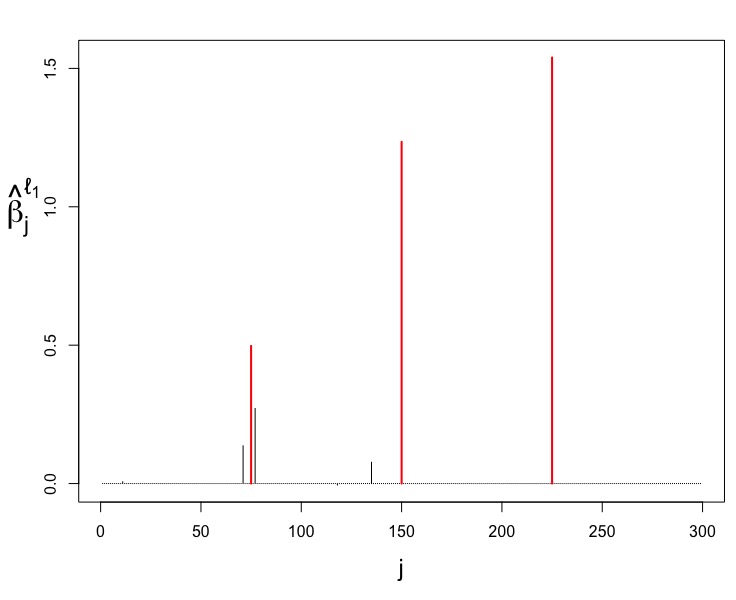}
		\caption{$(q, M) = (n, 30)$\label{seg_M30}}
	\end{subfigure}
	\caption{Coefficients of $\betahatlone$ when Lasso-Zero is tuned with different parameters $q$ and $M$, illustrating the effect of using no (i), one (ii) or several (iii) noise dictionaries. The red bars correpond to $\betahatlone_j$s with $j \in S^0.$}\label{TVexample}
\end{figure}

To illustrate the use of noise dictionaries, we consider an example of the simulation setting~(\ref{segX}) described in Section~\ref{settings}. Here $n=300$, $p=n-1,$ the true vector $\beta^0$ has support $S^0 = \{75, 150,  225\}$ and $\beta^0_j = 2.5$ for every $j \in S^0.$ Figure~\ref{TVexample} shows $\betahatlone$ (step \ref{lass0_before_thresh} in Definition \ref{deflass0}, that is, Lasso-Zero before being thresholded) for different values of $q$ and $M.$ The entries of~$\betahatlone_j$ for $j \in S^0$ are indicated in red. If $(q, M)=(0, 1)$ (Figure~\ref{seg_BP}), no noise dictionary is used and $\betahatlone$ is simply the solution to BP~(\ref{basis_pursuit}); the three largest coefficients in absolute value are not the ones indexed by~$S^0$, therefore exact support recovery after thresholding is impossible. In Figures~\ref{seg_M1} and~\ref{seg_M30}, noise dictionaries of size $q=n$ are used. Figure~\ref{seg_M1} illustrates the effect of using a single one ($M=1$): sparsity is induced in $\betahatlone$ since the columns of $G^{(1)}$ can also be used to fit $y$, but exact support recovery is still impossible. In Figure~\ref{seg_M30} with $M=30$ noise dictionaries, taking the median of the obtained $\betahat^{(k)}_j$, $k=1, \ldots, M$,  has the effect of setting almost all truly null coefficients to zero. With an appropriate threshold  (say $\tau = 0.4$ here), the estimated support equals~$S^0.$

Model selection is sometimes compared to finding needles in a haystack. By adding columns to the matrix $X,$ we add straw to the haystack and at first glance make the task of finding needles even more difficult. But the extra columns are \emph{known} to be straw, whereas in the original problem it is not known a priori which variables are the needles to find and which ones are the hay. Additionally, the associated noise coefficients~$\gammahat^{(k)}$ have the advantage of carrying some information about the noise level $\sigma$. This fact is exploited in Section~\ref{qut} to get a pivotal statistic allowing to select the threshold $\tau$ when $\sigma$ is unknown.

\section{Selecting the threshold} \label{thresh_selection}
Finding an appropriate threshold is crucial for model selection. A good model selection procedure should have high power, typically measured by the true positive rate ($\operatorname{TPR}$):
\begin{equation}\label{defTPR}
\operatorname{TPR} =  \E\left(\frac{\abs{S^0 \cap \Shat}}{\abs{S^0}} \right),
\end{equation}
and a controlled type I error. Common measures for the type I error include the family-wise error rate ($\operatorname{FWER}$) and the false discovery rate ($\operatorname{FDR}$), defined by
\begin{align} 
&\operatorname{FWER} = \p(\abs{\Shat \cap \comp{S^0}} > 0), \label{defFWER}
\\
&\operatorname{FDR} = \E \left( \frac{\abs{\Shat \cap \comp{S^0}}}{\abs{\Shat} \vee 1} \right). \label{defFDR}
\end{align}
They correspond to the probability of making at least one false discovery and the expected proportion of false discoveries among all discoveries, respectively. Since $\operatorname{FDR} \leq \operatorname{FWER}$, any procedure controlling the FWER at level $\alpha$ also controls the FDR. 
%

The threshold selection procedure suggested below is based on quantile universal thresholding (QUT), a general methodology  \citep{giacobino2017} which controls the $\operatorname{FWER}$ in the weak sense.
We also  exploit the noise coefficients $\gamma^{(k)}$ to yield a data dependent QUT that does not require any knowledge or preliminary estimation of the noise level $\sigma$. 

\subsection{Quantile universal thresholding} \label{qut}
Standard methods for selecting tuning parameters like cross-validation and other information criteria like SURE \citep{stein1981}, BIC \citep{schwarz1978}, AIC \citep{akaike1998}, are driven by prediction goals. Quantile universal thresholding, on the contrary, is directed towards model selection. It consists in controlling the probability of correctly recognizing the null model when $\beta^0 = 0$. In other words, if $\betahat_\lambda$ is a regularized estimator tuned by a parameter~$\lambda$, the $\alpha$-quantile universal threshold $\lambda_\alpha$ is the value of $\lambda$ for which $\p(\betahat_\lambda(\epsilon) = 0) = 1 - \alpha.$
 
By Definition~\ref{deflass0}, $\betahat^{\rm{lass0}}_\tau = 0$ if and only if $\tau \geq \norm{\betahatlone}_\infty,$ hence the $\alpha$-quantile universal threshold for Lasso-Zero is given by
\begin{equation} \label{lass0_QUT}
\tau_\alpha^{\operatorname{QUT}} = F^{-1}_T(1-\alpha),
\end{equation}
where $F_T$ denotes the c.d.f.~of 
\begin{equation} \label{statT}
T:=\norm{\betahatlone(\epsilon)}_\infty.
\end{equation}
If the value of $\sigma$ is known, the threshold~(\ref{lass0_QUT}) can be estimated by Monte Carlo simulations. However in practice the noise level is rarely known. In such cases it is often estimated beforehand, for example using a refitted cross-validation procedure \citep{fan2012} or based on the Lasso tuned by cross-validation like in \citet{reid2016}. It is rather suggested here to pivotize the statistic $\norm{\betahatlone(\epsilon)}_\infty$ using the noise coefficients $\gammahat^{(k)}.$ We consider
\begin{equation} \label{piv_stat}
P := \frac{\norm{\betahatlone(\epsilon)}_\infty}{s(\epsilon)},
\end{equation}
with 
$
s(y) := \operatorname{MAD}^{\neq 0}(\gammahat^{(1)}(y), \ldots, \gammahat^{(M)}(y)),
$
where $\operatorname{MAD}^{\neq 0}(v) := \operatorname{MAD}\{v_i \ \mid \ v_i \neq 0 \}.$ 
The numerator and denominator in~(\ref{piv_stat}) being both equivariant with respect to $\sigma$, the distribution $F_P$  does not depend on $\sigma$. 
It is then straightforward to check that choosing the data-dependent quantile universal threshold \begin{equation} \label{lass0_pivQUT}
\tau_\alpha^{\operatorname{QUT}}(y) := s(y) F_P^{-1}(1-\alpha)
\end{equation}
ensures $\p(\betahat^{\rm{lass0}}_\tau(y) = 0 \ \mid \tau=\tau_\alpha^{\operatorname{QUT}}(y)) = 1-\alpha$ 
under $\beta^0 = 0$. 

\subsection{Approximation of QUT based on a GEV fit} \label{GEV}
Since the distribution of $P$~(\ref{piv_stat}) is unknown, estimating the quantile universal threshold~(\ref{lass0_pivQUT}) requires to perform a Monte Carlo simulation to obtain an estimate $\hat{q}_\alpha^{\textrm{MC}}$ of the upper $\alpha$-quantile $q_\alpha := F_P^{-1}(1-\alpha)$. A large number of Monte Carlo replications are needed for good estimation, especially for small values of $\alpha.$ This can be expensive for large design matrices since the extended BP problem~(\ref{strawlass0}) needs to be solved $M$ times for a single realization of $P$. Parallel computing allows a great gain of time if one has the computational resources. But if one can afford only a small number of replications, a parametric fit could improve the approximation. Indeed $P$ is the maximum of $p$ random variables, hence could be modelled by a GEV distribution \citep{embrechts1997} which can be fitted using a relatively small number of realizations of $P.$ Once the GEV distribution is fitted, its upper $\alpha$-quantile $\hat{q}_\alpha^{\textrm{GEV}}$ can be used as an estimate of $q_\alpha.$

\newcommand{\GEVheightone}{5cm}
\newcommand{\GEVheight}{4cm}
\newcommand{\GEVscale}{0.2}
\renewcommand{\thesubfigure}{(\alph{subfigure})}
\begin{figure}[!ht]
	\centering
	\begin{subfigure}[t][\GEVheightone]{6cm}
	\centering
	\includegraphics[scale=\GEVscale]{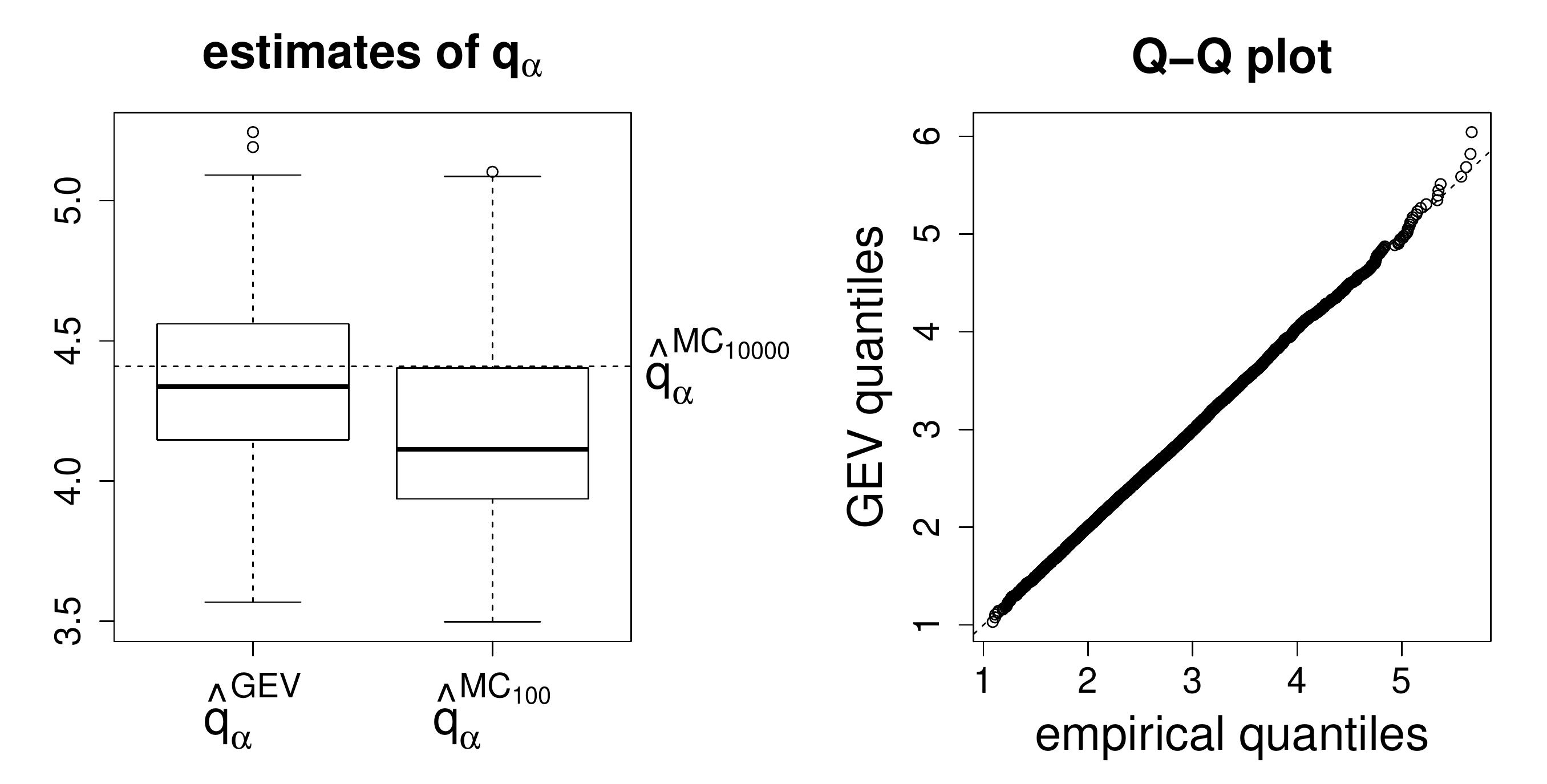}
	\caption{}
	\end{subfigure}
	\qquad
	\begin{subfigure}[t][\GEVheightone]{6cm}
	\centering
	\includegraphics[scale=\GEVscale]{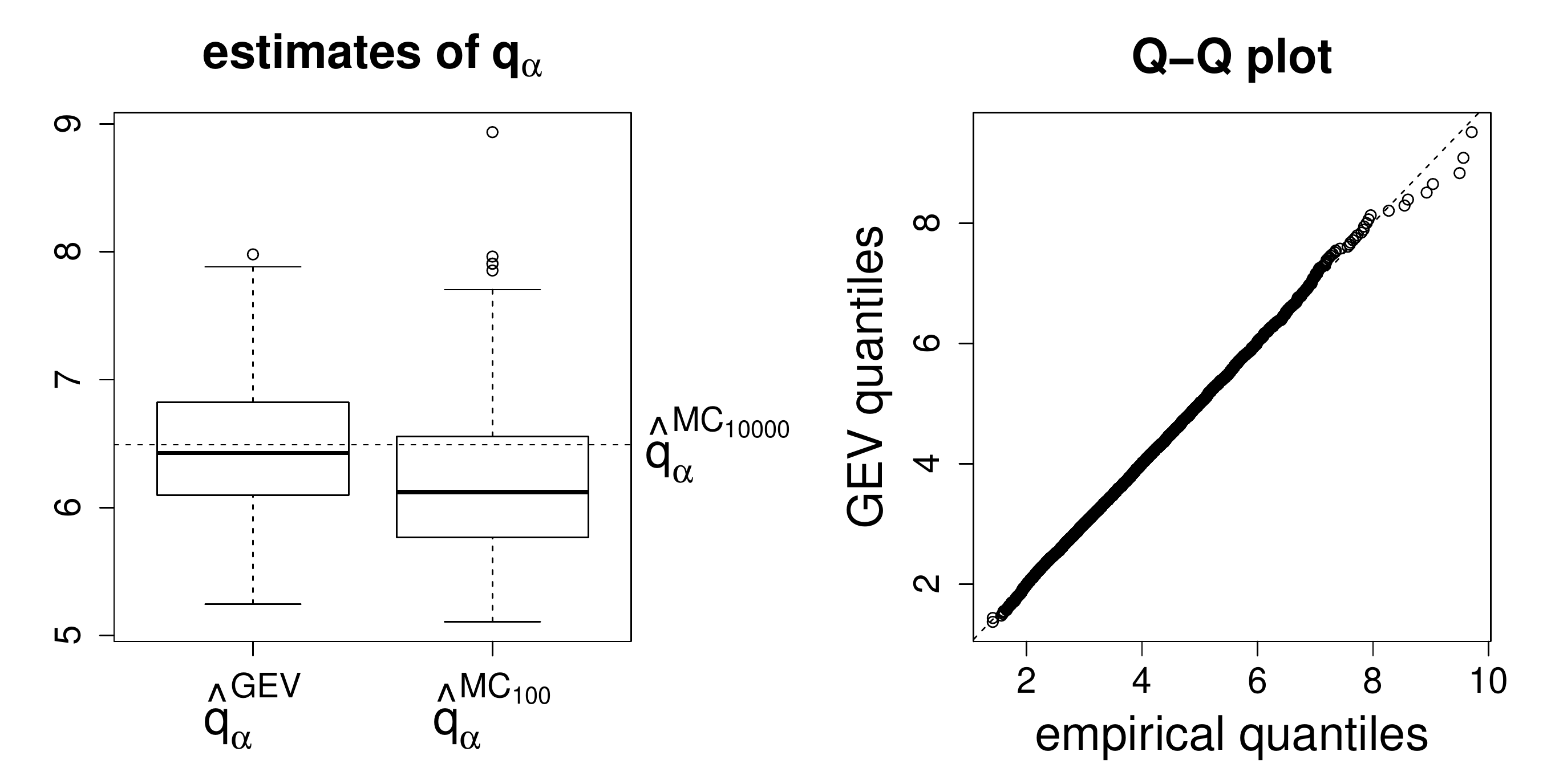}
		\caption{}
	\end{subfigure}
	\begin{subfigure}[t][\GEVheight]{6cm}
	\centering
	\includegraphics[scale=\GEVscale]{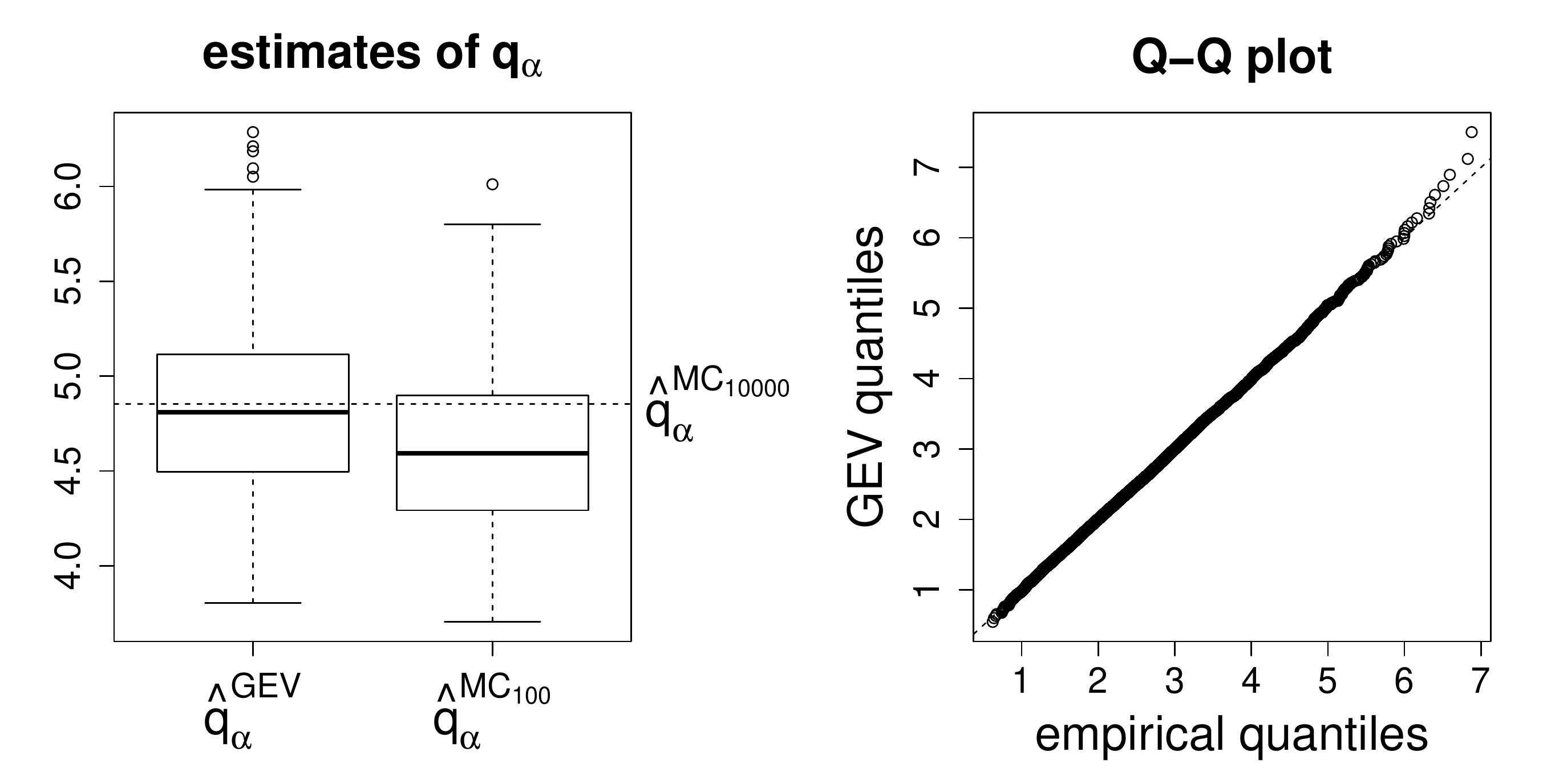}
		\caption{}
	\end{subfigure}
	\qquad
	\begin{subfigure}[t][\GEVheight]{6cm}
	\centering
	\includegraphics[scale=\GEVscale]{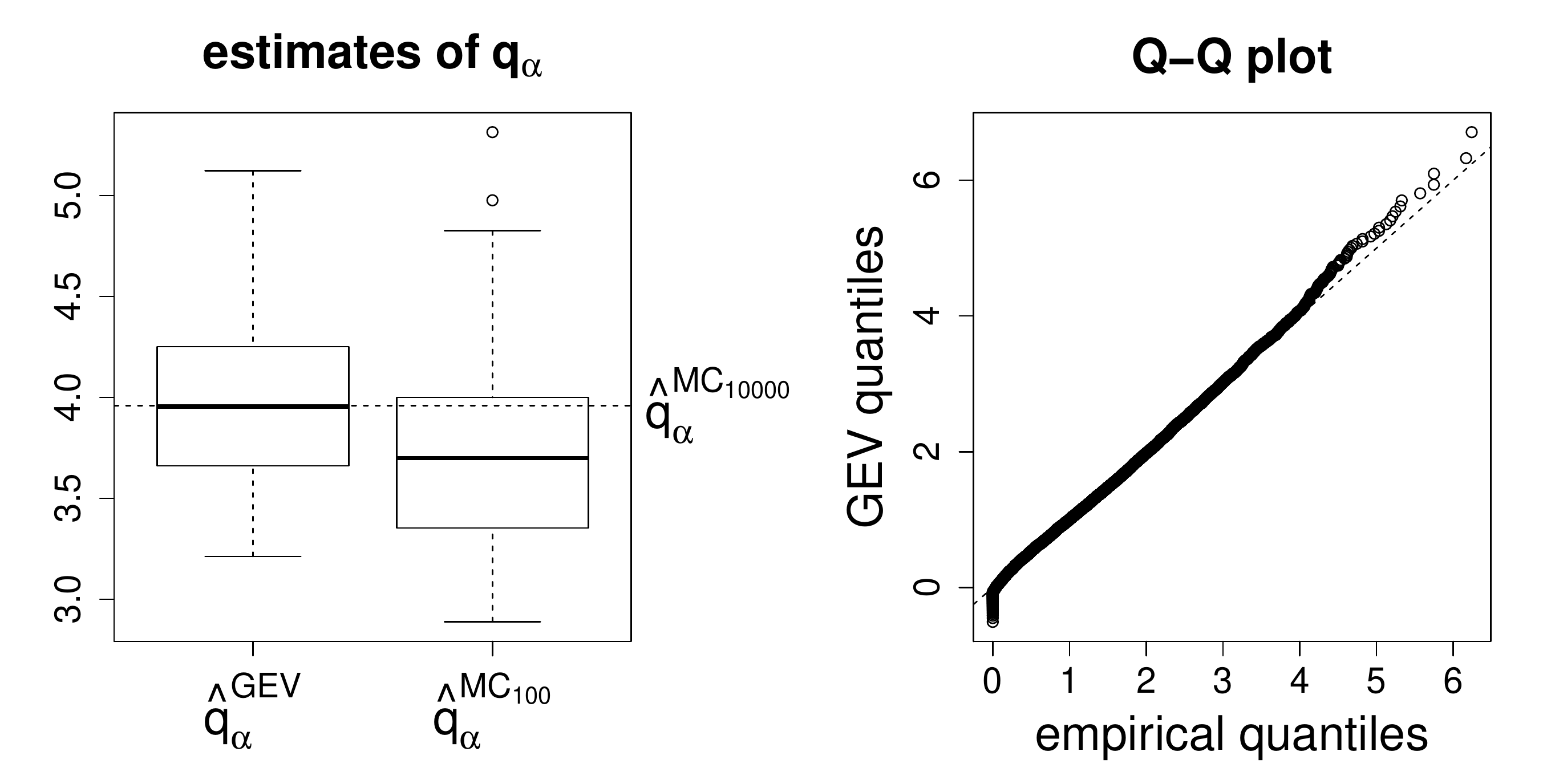}
	\caption{}
	\end{subfigure}
	
	\caption{Approximation of $q_\alpha := F_P^{-1}(1-\alpha)$ in the four simulation settings~(\ref{smallGauss}),~(\ref{wideGauss}),~(\ref{ribo}),~(\ref{segX}) described in Section~\ref{settings}. Left plots: for 100 Monte Carlo replications, boxplots of 200 estimates $\hat{q}^{\textrm{GEV}}_\alpha$ based on a GEV fit and 200 empirical estimates $\hat{q}^{\textrm{MC}_{100}}_\alpha.$ The empirical estimate based on $10,000$ replications is represented by the horizontal line. Right plots: Q-Q plot of the GEV fit against the empirical quantiles of the $10,000$ sample. \label{GEVplots}}
\end{figure}

In our numerical experiments of Section~\ref{simus}, $q_\alpha$ was estimated by a Monte Carlo simulation with $10,000$ replications. To investigate the success of the GEV approach when a smaller number of replications are made and $\alpha$ is small (we consider $\alpha = 0.01$ in this section), subsamples of size $100$ were repeatedly drawn among our $10,000$ realizations of~$P.$ For each subsample we computed the upper $\alpha$-quantile $\hat{q}^{\text{GEV}}_\alpha$ of the GEV distribution fitted by maximum likelihood and the empirical upper $\alpha$-quantile $\hat{q}^{\text{MC}_{100}}_\alpha$ obtained with this subsample. If the law of $P$ can be well approximated by a GEV distribution, one can expect~$\hat{q}^{\text{GEV}}_\alpha$ to be closer to $q_\alpha$ than $\hat{q}^{\text{MC}_{100}}_\alpha$ since the tail of the distribution cannot be well approximated with a small sample and the empirical quantile is biased downwards. Of course the exact~$q_\alpha$ value is unknown, so instead we use the empirical quantile $\hat{q}^{\text{MC}_{10000}}_\alpha$ computed with all of our $10,000$ realizations as a reference. Figure~\ref{GEVplots} shows the result of this comparison in the four simulation settings described in Section~\ref{settings}. The left plots of all subfigures show the boxplots for $\hat{q}^{\text{GEV}}_\alpha$ and $\hat{q}^{\text{MC}_{100}}_\alpha,$ compared to the Monte Carlo estimate $\hat{q}^{\text{MC}_{10000}}_\alpha$ indicated by the horizontal line. The right  Q-Q plots of the 
GEV fit against the empirical quantiles of the $10,000$ realizations  show that the GEV distribution approximates the law of $P$ well.

\section{Numerical experiments} \label{simus}
This section presents some simulations comparing Lasso-Zero to other variable selection procedures, looking at the $\operatorname{FDR}$~(\ref{defFDR}), $\operatorname{TPR}$~(\ref{defTPR}) and probability of exact support recovery. 

\subsection{Simulation settings} \label{settings}
We describe here the four different design matrices used in our numerical experiments.
\begin{enumerate}[(a)]
\item \label{smallGauss} \textit{i.i.d.~Gaussian design}: here $n=100$, $p=200$ and the entries $X_{ij}$ are i.i.d.~$N(0, 1).$ The columns of $X$ are then mean-centered and standardized so that $\sum_{i=1}^n X_{ij} = 0$ and $\frac{1}{n-1}\sum_{i=1}^n X_{ij}^2 = 1$ for every $j$. The matrix is generated once for the whole simulation. The amplitude of the nonzero coefficients is set to $0.75$ and their signs are random.
\item \label{wideGauss} \textit{wider i.i.d.~Gaussian design}: identical to~(\ref{smallGauss}) except $p=1000.$
\item \label{ribo} \textit{wide design with correlated variables (real dataset):} the $X$ matrix from the \texttt{riboflavin} dataset \citep{buhlmann2014} measuring the expression levels of $p=4088$ genes for $n=71$ Bacillus subtilis bacteria is characterized by a large ratio $p/n$ and  high correlations between variables. The matrix is mean-centered and standardized, 
and the nonzero coefficients are set to  $2$ with random signs. We also considered wide Gaussian matrices with Toeplitz covariance~$\Sigma_{ij} = 0.9^{\abs{i-j}}$, and obtained similar results.
\item \label{segX} \textit{highly correlated columns (segmentation problem):} here $n=300$ and the considered matrix is $X \in \R^{n \times (n-1)}$ defined by $X_{ij} = \bm{1}_{\{i > j\}}.$ It is related to the problem of recovering the jumps of a piecewise constant signal $f^0$ of size $n$, given a noisy observation $y= f^0 + \epsilon$. Indeed, this problem falls within the sparse linear framework by considering the vector $\beta^0 \in \R^{n-1}$ of all successive differences $\beta^0_j = f^0_{j+1} - f^0_j$ and writing  $
f^0 = f^0_1 \bm{1}_n + X \beta^0$,
where $\bm{1}_n \in \R^n$ is the vector of ones. Then recovering the discontinuity points of $f^0$ is equivalent to recovering the support of $\beta^0$. The discontinuities have amplitude $3$ and random signs. To omit the intercept $f^0_1$, both the observation $y$ and the matrix $X$ are mean-centered, so $\beta^0$ is estimated from the response $\tilde{y} = (I -P_{\bm{1}_n})y$ and the matrix $\tilde{X} = (I -P_{\bm{1}_n})X,$ where $P_{\bm{1}_n}$ is the orthogonal projection matrix onto the subspace generated by $\bm{1}_n.$ The Lasso is equivalent to total variation  \citep{rudin1992}
\[
\hat{f}^{\rm{TV}} = \argmin_{f \in \R^n} \frac{1}{2}\norm{y - f}_2^2 + \lambda \sum_{j=1}^{n-1} \abs{f_{j+1} - f_j},
\]
in the sense that $\hat{f}^{\rm{TV}} = \hat{f}^{\rm{TV}}_1 + X \betahat^{\rm{lasso}}_\lambda.$
\end{enumerate}

In all the simulations the parameters $(q, M)$ of Lasso-Zero are set to $q=n$ and $M=30$ and the threshold is selected as in~(\ref{lass0_pivQUT}) with $\alpha = 0.05$. The noise level is set to $\sigma = 1$ and the support $S^0$ is chosen randomly for each realization, except in setting~(\ref{segX}) where the discontinuity points are always equidistant. The sparsity index $s^0 := \abs{S^0}$ starts from~$0$ (null model) and is increased until exact support recovery becomes almost impossible.

We make here an important clarification regarding the standardization of the matrices $X$ and $G^{(k)}.$ It is desirable that the vectors of the noise dictionaries are on the same scale as the columns of $X$. So when $X$ is standardized so that the empirical standard deviation of each column $X_j$ equals 1, the same is done on the noise dictionaries $G^{(k)}.$ In setting~(\ref{segX}), the matrix $\tilde{X}$ is not standardized in order to preserve the particular structure of the problem. In this case, all columns of $G^{(k)}$ are standardized to have the same $\ell_2$-norm in such a way that the upper $\alpha$-quantiles of the statistics $\norm{\tilde{X}^T \epsilon}_\infty$ and $\norm{(G^{(k)})^T \epsilon}_\infty$ coincide.

\subsection{Choice of competitors} \label{competitors}

In all simulation settings, we computed the Lasso solution tuned by GIC \citep{fan2013}. Additionnally to the Lasso, we considered  SCAD \citep{fan2004} tuned by the GIC criterion (computed with the \texttt{R} package \texttt{ncvreg}~\citep{breheny2011}), which uses a nonconvex penalty improving the model selection performance over the Lasso. 

We also selected two procedures aiming at controlling the FDR, namely SLOPE~\citep{bogdan2015} and the knockoff filter \citep{barber2015, candes2018}, computed with the \texttt{R} packages \texttt{SLOPE}~\citep{RSLOPE} and \texttt{knockoff}~\citep{Rknockoff} respectively. SLOPE provably controls the FDR in the orthogonal case only, and the authors proposed a heuristic sequence of tuning parameters that are appropriate for independent Gaussian designs. As it is not known how to tune it for other design matrices, SLOPE was computed in settings~(\ref{smallGauss}) and~(\ref{wideGauss}) only. The knockoff filter was not tuned to ``knockoff+", which provably controls the FDR, since it is very conservative and has no power in our simulation settings. So we chose the option \texttt{offset=0} so that the knockoff filter only controls a modified version of the FDR and leads to less conservative results. Model-X knockoffs \citep{candes2018} were used in settings~(\ref{smallGauss}) and~(\ref{wideGauss}), together with the statistic \texttt{stat.glmnet\_coefdiff} (based on the comparison of the coefficients to their respective knockoffs at the Lasso solution tuned by cross-validation), wheareas the knockoff filter for fixed (low-dimensional) design \citep{barber2015} was used in setting~(\ref{segX}) together with the statistic \texttt{stat.lasso\_lambdasmax}. It was excluded in setting~(\ref{ribo}) due to computational time. Both SLOPE's and knockoffs' FDR target was set to $0.05$.

Since SLOPE and knockoffs do not actually control the FDR in our simulations, we decided to include stability selection~\citep{meinshausen2010} (computed with the package \texttt{stabs}~\citep{Rstabs}
with \texttt{cutoff=0.6} and \texttt{PFER=1} for comparable FDR), another conservative procedure designed to control the FWER.

In order to compare our results to state-of-the-art methods in the segmentation problem~(\ref{segX}), we also ran wild binary segmentation \citep{fryzlewicz2014}, which was specifically introduced for segmentation, using the \texttt{R} package \texttt{wbs} \citep{Rwbs}.

Finally, for estimators requiring an estimation of the noise level, namely Lasso, SCAD (both tuned by GIC) and SLOPE, the value of $\sigma^2$ was estimated based on the residual sum of squares for the Lasso tuned by cross-validation, as in \citet{reid2016} (the current implementation of SLOPE to estimate $\sigma$ ran into occasional convergence problems under our simulation regimes). In the segmentation problem~(\ref{segX}), the estimation of $\sigma$ was based on the $\operatorname{MAD}$ of the successive differences of $y$.

\subsection{Results}
\newcommand{\fdrscale}{0.3}
\newcommand{\figwidth}{6cm}
\newcommand{\fdrheightone}{7cm}
\newcommand{\fdrheight}{6cm}
\renewcommand{\thesubfigure}{(\alph{subfigure})}
\begin{figure}[!ht]
	\centering
	\begin{subfigure}[t][\fdrheightone]{\figwidth}
	\centering
	\includegraphics[scale=\fdrscale]{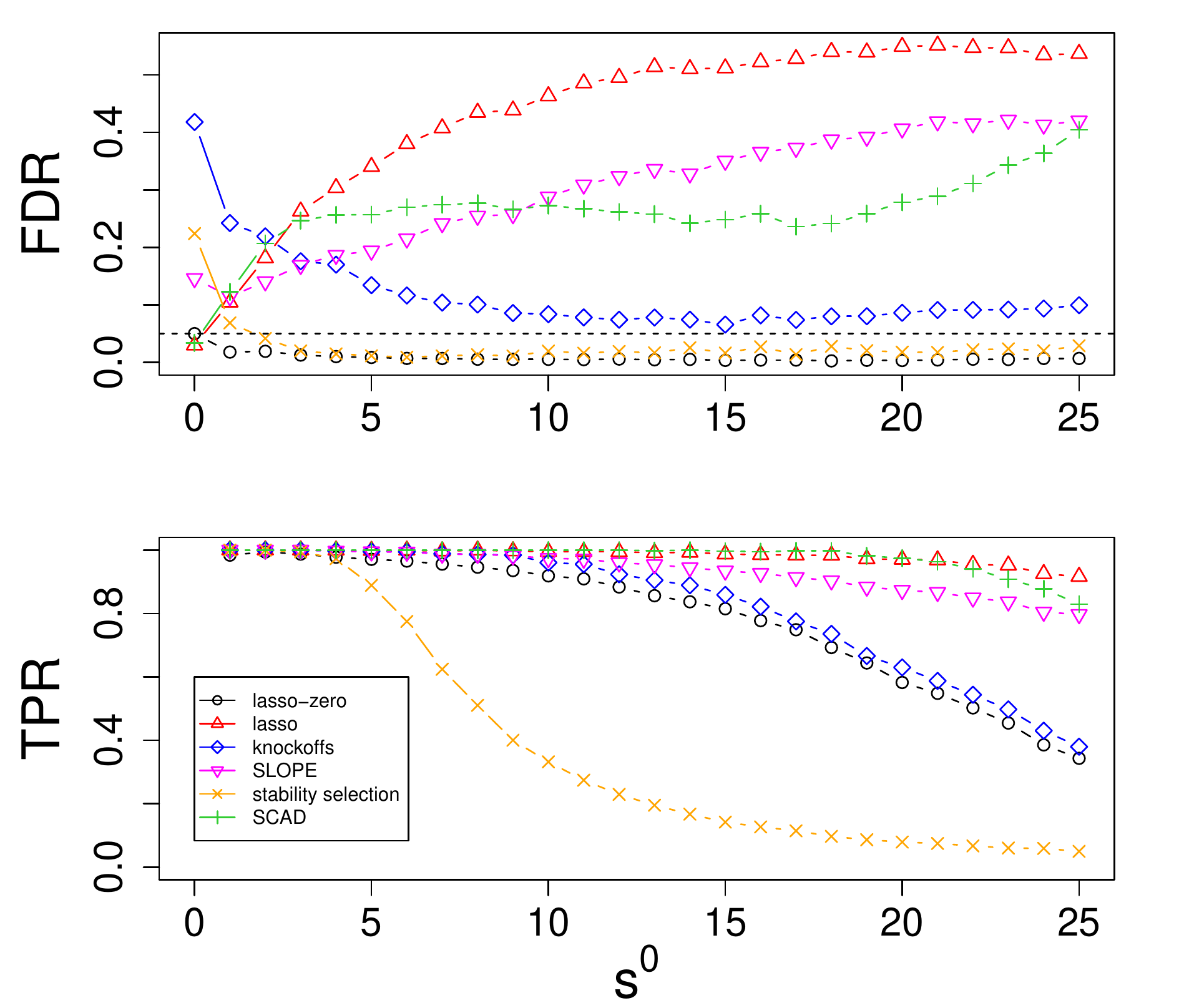}
	\caption{}
	\end{subfigure}
	\qquad
	\begin{subfigure}[t][\fdrheightone]{\figwidth}
	\centering
	\includegraphics[scale=\fdrscale]{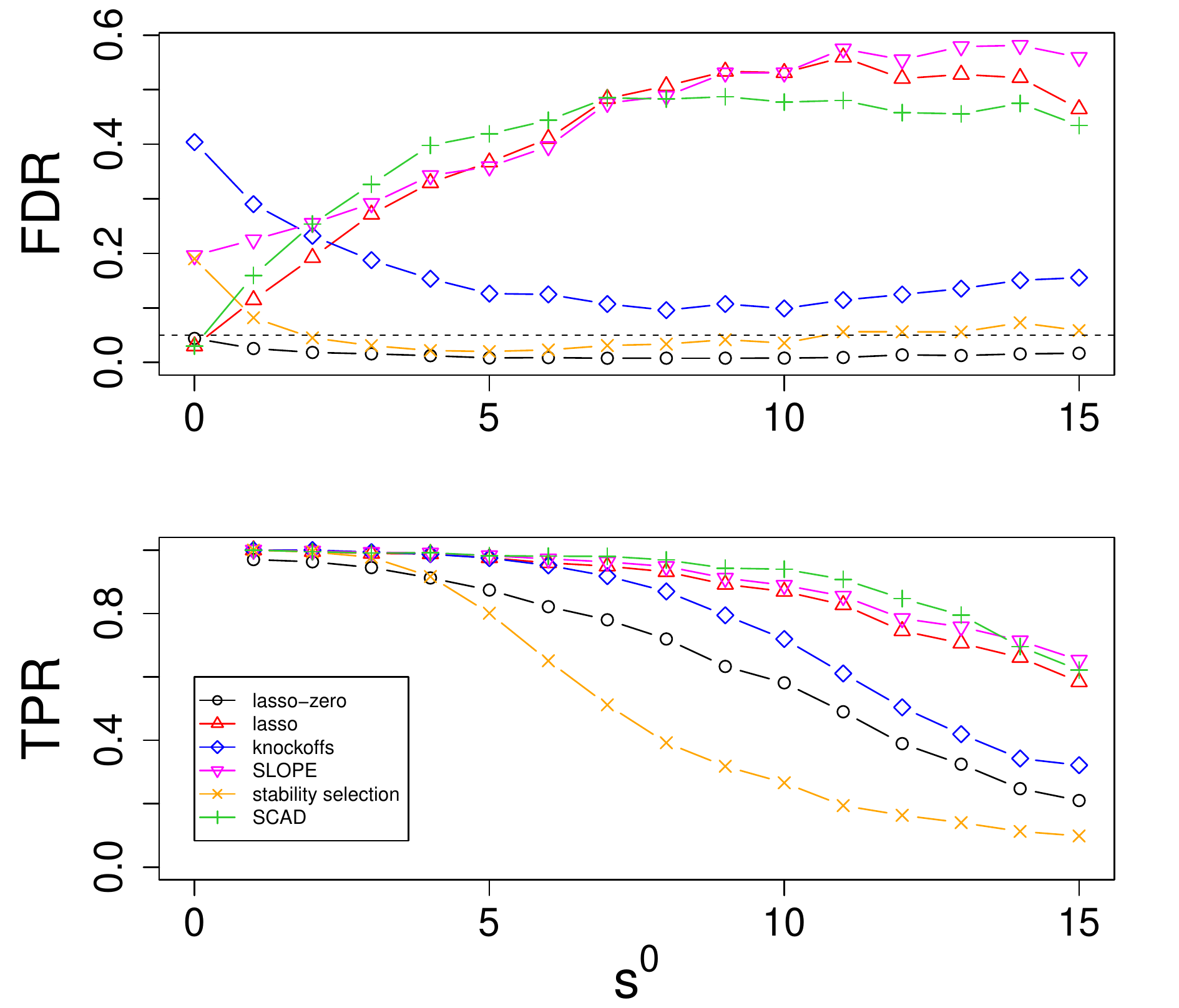}
		\caption{}
	\end{subfigure}
	\begin{subfigure}[t][\fdrheight]{\figwidth}
	\centering
	\includegraphics[scale=\fdrscale]{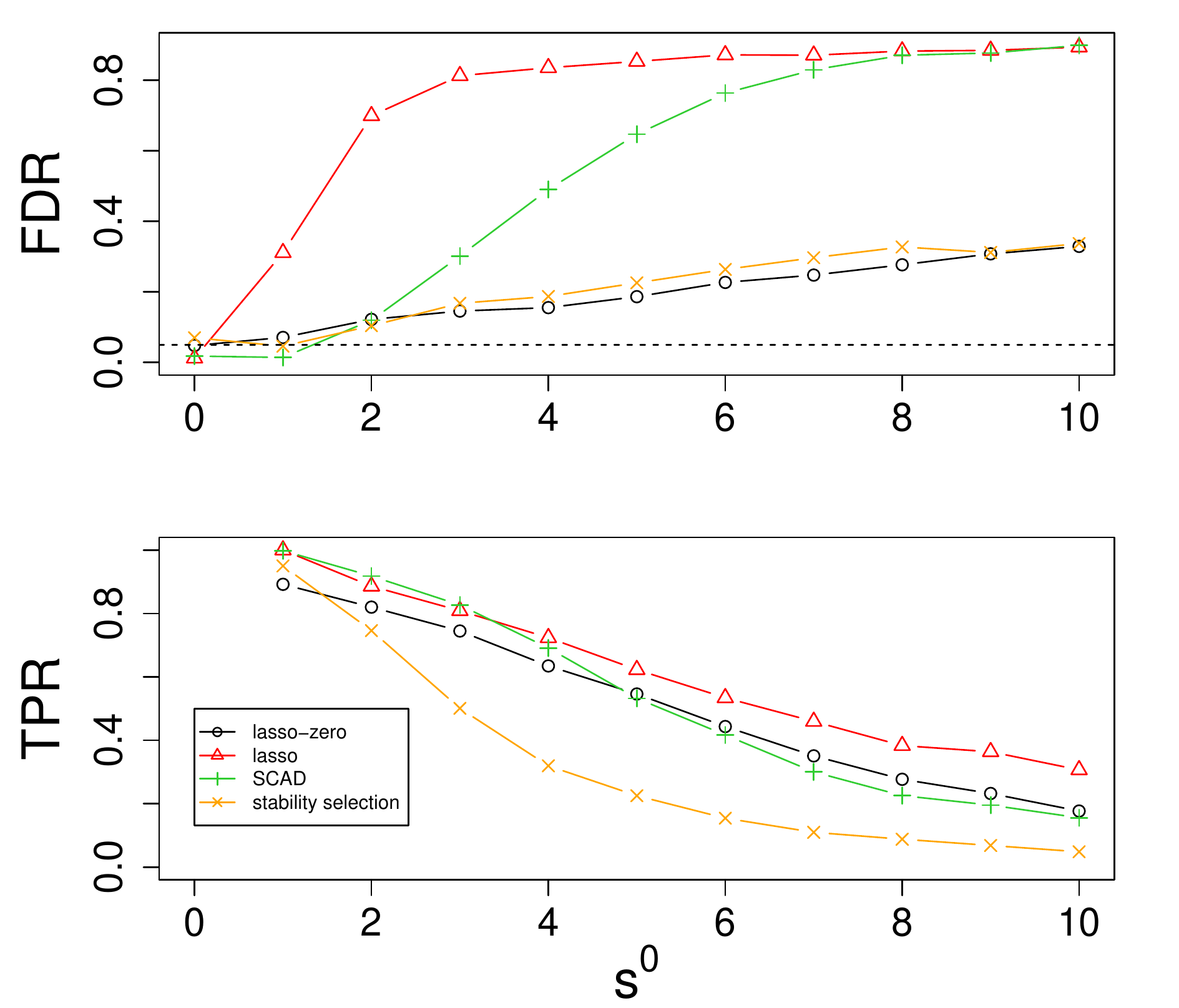}
		\caption{}
	\end{subfigure}
	\qquad
	\begin{subfigure}[t][\fdrheight]{\figwidth}
	\centering
	\includegraphics[scale=\fdrscale]{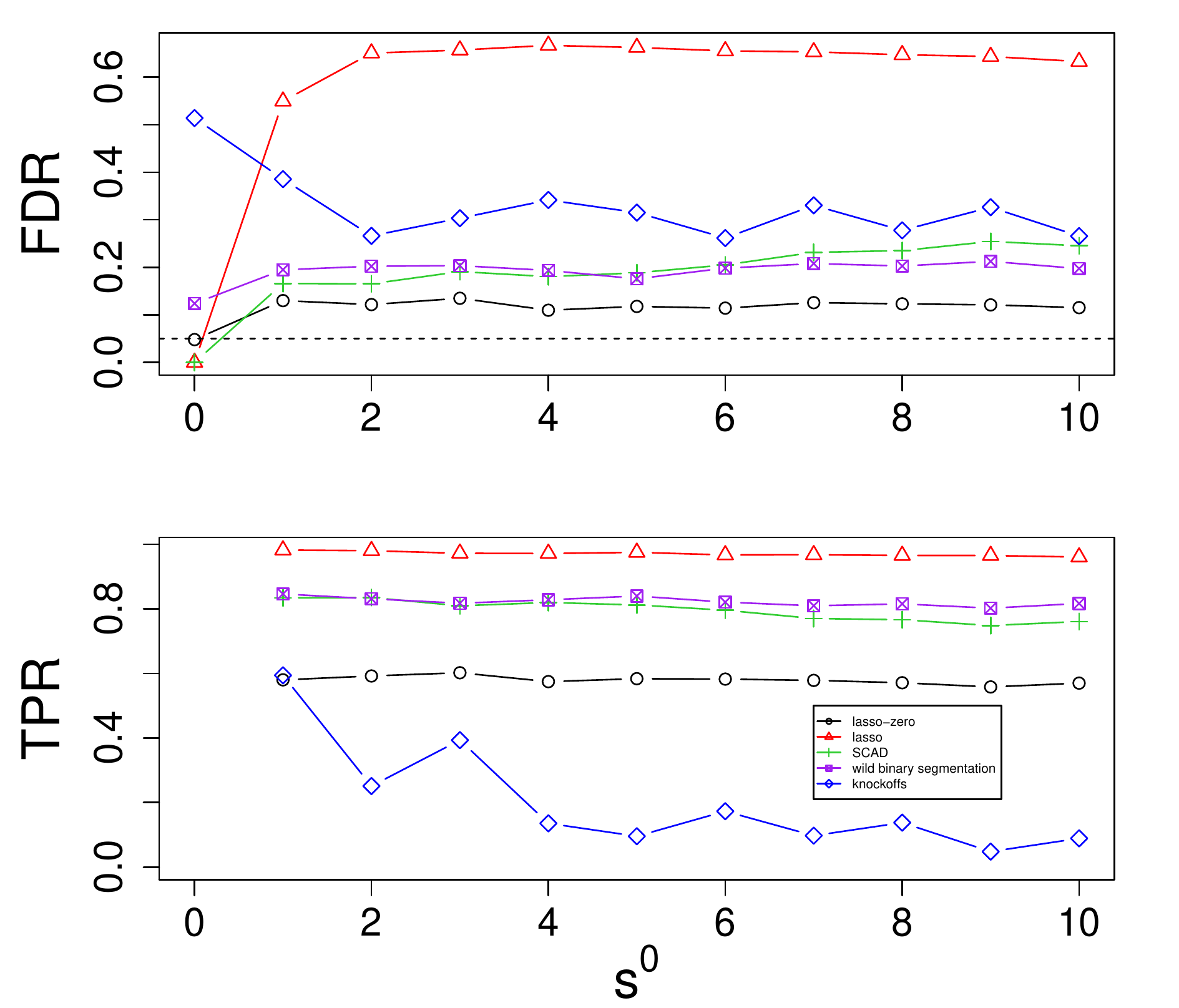}
	\caption{}
	\end{subfigure}
	
	\caption{FDR and TPR 
	 for settings~(\ref{smallGauss}),~(\ref{wideGauss}),~(\ref{ribo}),~(\ref{segX}) of Section~\ref{settings}.\label{tprfdr_graph}}
\end{figure}

Figure~\ref{tprfdr_graph} shows the $\operatorname{FDR}$~(\ref{defFDR}) and $\operatorname{TPR}$~(\ref{defTPR}) after $500$ replications as the model size $s^0$ grows. Lasso-Zero's $\operatorname{FDR}$ is the lowest in all simulation settings. By construction of the quantile universal threshold, it is equal to $\alpha = 0.05$ when $s^0 = 0.$ Interestingly, we observe that in the i.i.d.~Gaussian settings~(\ref{smallGauss}) and~(\ref{wideGauss}) the $\operatorname{FDR}$ remains controlled at level $\alpha$ even as $s^0$ increases. This phenomenon no longer holds in the correlated cases~(\ref{ribo}) and (\ref{segX}).

Compared to stability selection (in settings~(\ref{smallGauss}), (\ref{wideGauss}) and (\ref{ribo})), which achieves a comparable $\operatorname{FDR}$, Lasso-Zero is much more powerful, as shown by a higher $\operatorname{TPR}$ value. Lasso, SLOPE and SCAD are more powerful than Lasso-Zero, but at the price of a much higher $\operatorname{FDR}.$ So compared to other estimators, Lasso-Zero achieves an excellent trade-off between low $\operatorname{FDR}$ and high $\operatorname{TPR}$, except in the segmentation setting~(\ref{segX}) where Lasso-Zero's power is clearly lower than Lasso, SCAD and wild binary segmentation, however without improving much the $\operatorname{FDR}.$ In this case, Lasso-Zero's performance remains however better than the knockoff filter, which presents a high $\operatorname{FDR}$ and low $\operatorname{TPR}.$

\newcommand{\suppwidth}{0.28}
\newcommand{\suppheightone}{6cm}
\newcommand{\suppheight}{5cm}
\renewcommand{\thesubfigure}{(\alph{subfigure})}
\begin{figure}[!ht]
	\centering
	\begin{subfigure}[t][\suppheightone]{6cm}
	\centering
	\includegraphics[scale=\suppwidth]{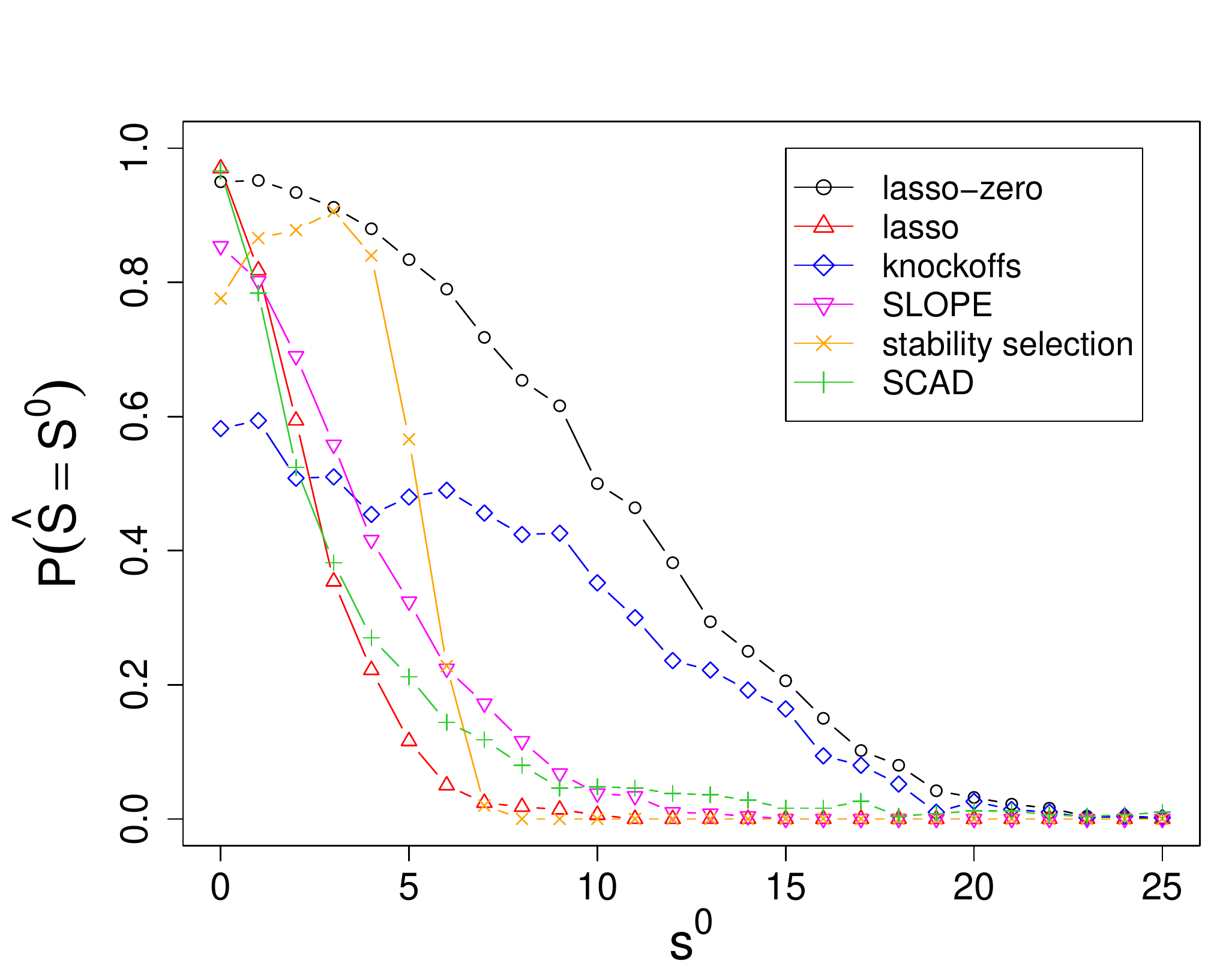}
	\caption{}
	\end{subfigure}
	\qquad
	\begin{subfigure}[t][\suppheightone]{6cm}
	\centering
	\includegraphics[scale=\suppwidth]{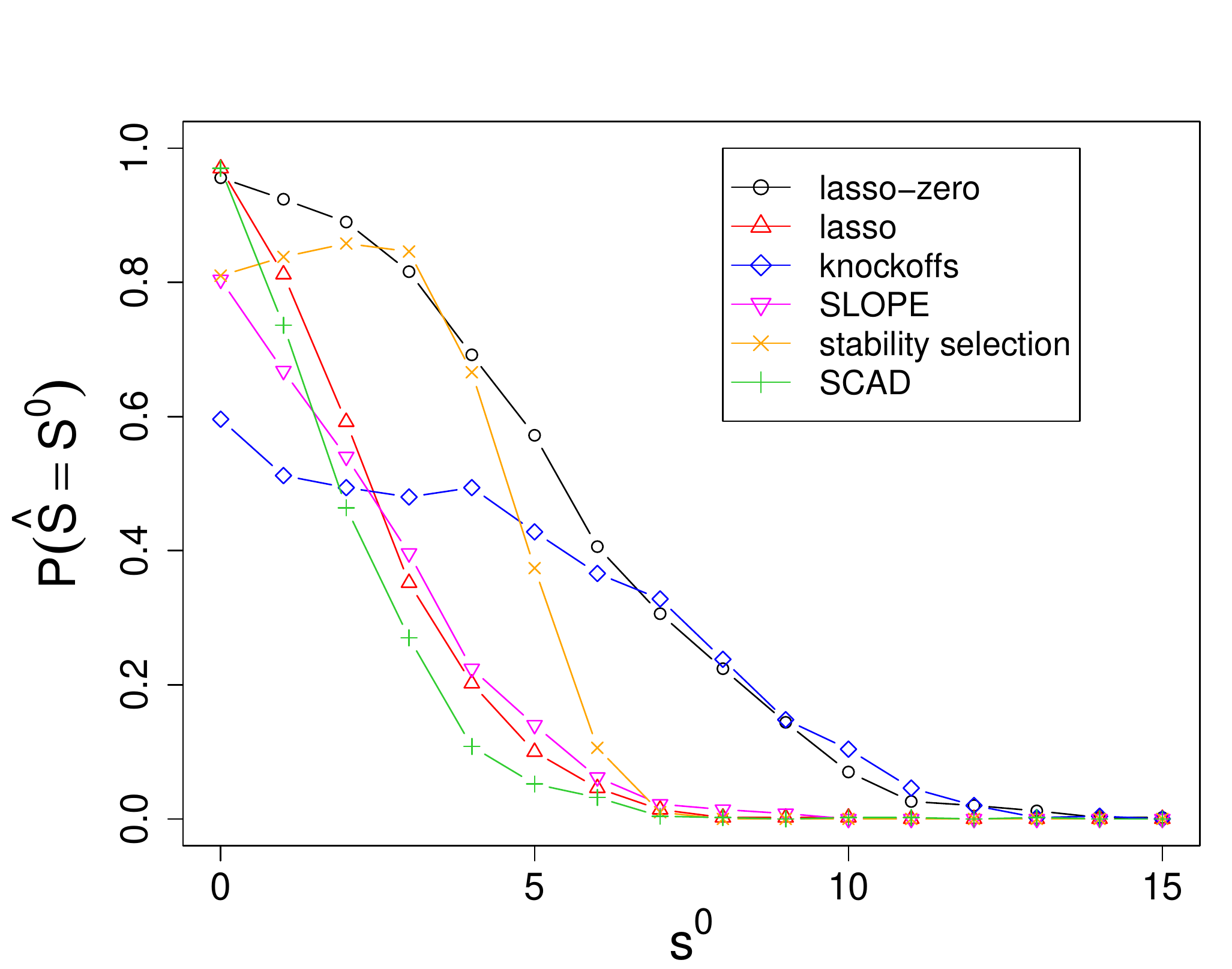}
		\caption{}
	\end{subfigure}
	\begin{subfigure}[t][\suppheight]{6cm}
	\centering
	\includegraphics[scale=\suppwidth]{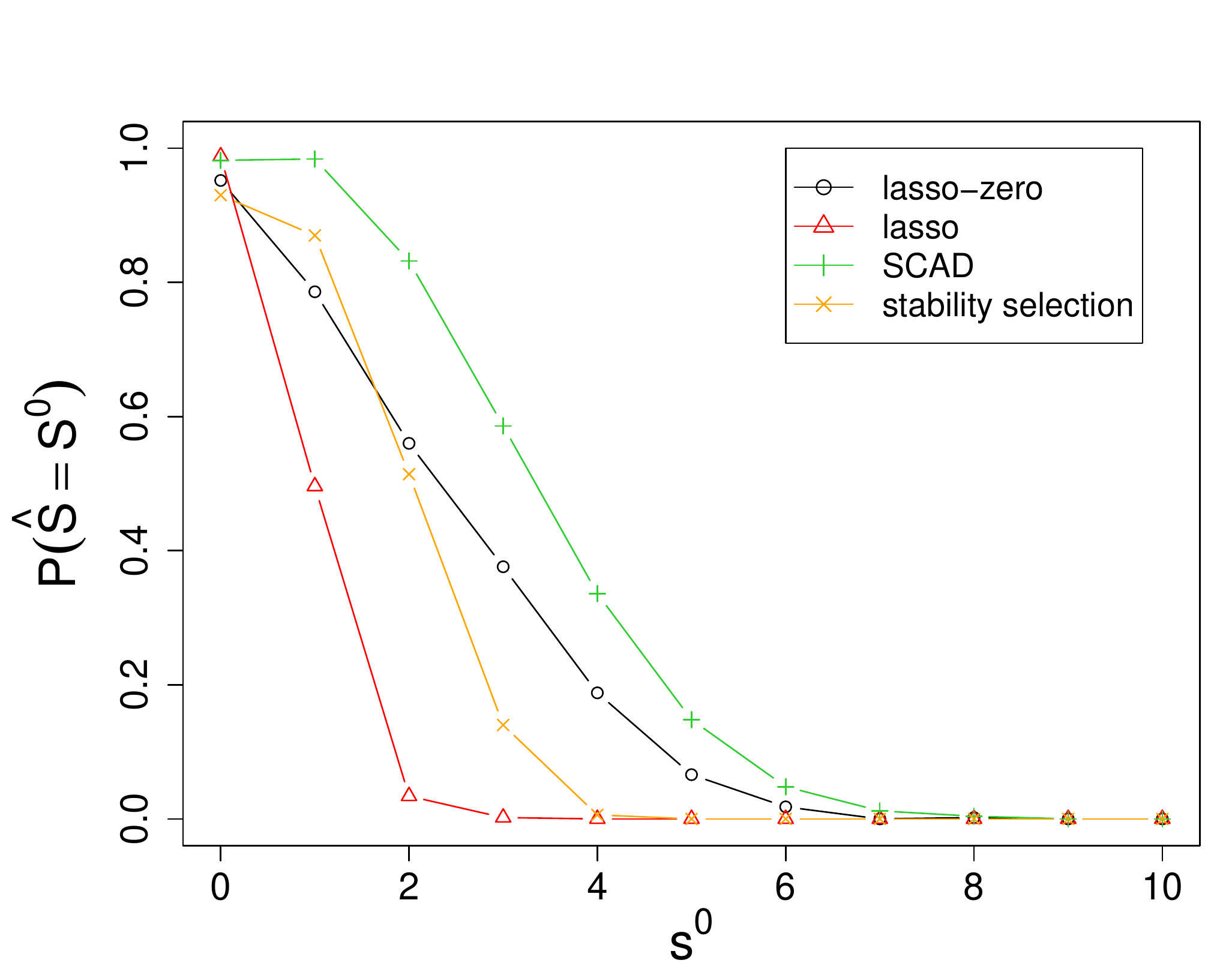}
		\caption{}
	\end{subfigure}
	\qquad
	\begin{subfigure}[t][\suppheight]{6cm}
	\centering
	\includegraphics[scale=\suppwidth]{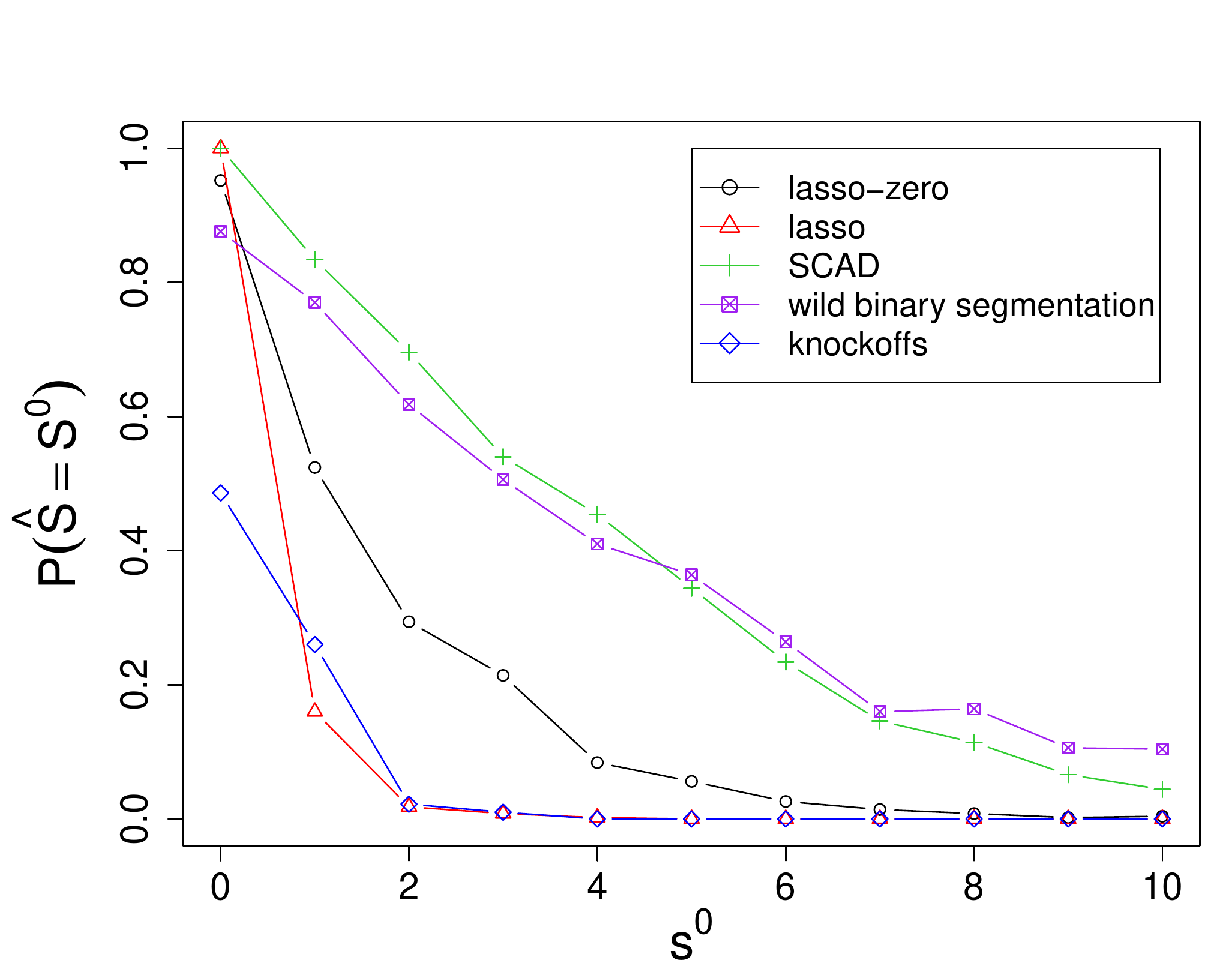}
	\caption{}
	\end{subfigure}
	
	\caption{Probability of exact support recovery for settings~(\ref{smallGauss}),~(\ref{wideGauss}),~(\ref{ribo}),~(\ref{segX}) of Section~\ref{settings}.\label{supp_graph}}
\end{figure}

Figure~\ref{supp_graph} shows the frequency of exact support recovery after $500$ replications. Lasso-Zero performs better than its competitors in the Gaussian settings~(\ref{smallGauss}) and (\ref{wideGauss}), and is outperformed only by SCAD in setting~(\ref{ribo}). In the segmentation setting~(\ref{segX}), Lasso-Zero represents a clear improvement over the Lasso and knockoffs, but does not perform as well as wild binary segmentation and SCAD. This is all consistent with the results shown in Figure~\ref{tprfdr_graph}, since Lasso-Zero's $\operatorname{TPR}$-$\operatorname{FDR}$ trade-off is particularly good in~(\ref{smallGauss}) and (\ref{wideGauss}). Its slight $\operatorname{FDR}$ improvement in setting~(\ref{segX}) is made at the price of an important loss of power.

%
%
%

\section{Theoretical results} \label{theory}

A particular case of Lasso-Zero is when $q=0$ and $M=1$, i.e. when no random column is added to $X$. Then Lasso-Zero simply thresholds the solution to BP~(\ref{basis_pursuit}). This particular case is at the core of the Lasso-Zero idea and the results obtained in this case support our methodology. We focus here on this special case and throughout this section ``Lasso-Zero" always refers to $\betahat^{\rm{lass0}}_\tau = \betahat^{\operatorname{lass0}(0, 1)}_\tau,$ so
$
\betahat^{\rm{lass0}}_\tau = \eta_\tau(\betahat^{\ell_1}),
$
where $\betahat^{\ell_1}$ is the BP solution.


\subsection{Analysis under the stable null space property} \label{subsection:analysis_SNSP}
Our main finding is that Lasso-Zero requires a weaker condition on $X$ and~$S^0$ than the Lasso for exact recovery of~$S^0.$ To see that, let $\betahat^{\ell_1}(\epsilon)$ denote the solution to
\begin{equation} \label{opt:BP_epsilon}
\begin{aligned}
& \min_{\beta \in \R^p} \norm{\beta}_1 \quad  \text{ s.t.}  \quad  \epsilon = X \beta,
\end{aligned}
\end{equation}
in other words $\betahatlone(\epsilon)$ is the BP solution obtained if we were to observe only the noise~$\epsilon.$ The smaller $\sigma$ is, the smaller the coefficients of~$\betahat^{\ell_1}(\epsilon)$ tend to be. Since $\epsilon = X \betahat^{\ell_1}(\epsilon),$ the linear model~(\ref{lin_model}) can be rewritten as
$
y = X (\beta^0 + \betahat^{\ell_1}(\epsilon))
$,
which we interpret as a noiseless problem in which it is desired to recover $\beta^0 + \betahat^{\ell_1}(\epsilon).$ If the perturbation $\betahat^{\ell_1}(\epsilon)$ is small we expect 
 that the BP solution $\betahat^{\ell_1}$ will be close to $\beta^0 + \betahat^{\ell_1}(\epsilon),$ so that after thresholding its coefficients at some level $\tau > 0$ only the components indexed by $j \in S^0$ remain nonzero. BP is efficient for recovering sparse vectors in the noiseless case. Although close to $\beta^0$, the vector $\beta^0 + \betahat^{\ell_1}(\epsilon)$ is less sparse. We therefore assume that $X$ and $S^0$ satisfy the \emph{stable null space property}, which ensures the stability of BP with respect to sparsity defect \citep{foucart2013}. The stable null space property is met with constant $\rho \in (0, 1)$ if
\begin{equation} \label{SNSP}
\norm{\beta_{S^0}}_1 \leq \rho \norm{\beta_{\comp{S^0}}}_1 \quad \text{for all } \beta \in \ker(X).
\end{equation}
Under the stable null space property, Lasso-Zero can recover the signs of $\beta^0$ exactly.

\begin{theorem} \label{main_thm}
Assume $\rank{X} = n$ and the stable null space property~(\ref{SNSP}) is met for $X$ and $S^0$ with constant $\rho \in (0, 1).$ Let $\beta^0_{\min} := \min\{\abs{\beta^0_j} \ \mid \ j \in S^0 \}.$ If
\begin{equation} \label{beta_min}
\beta^0_{\min} > C_{\rho} \norm{\betahatlone(\epsilon)}_1,
\end{equation}
where $C_{\rho} = \frac{2(3 + \rho)}{1 - \rho},$ then there exists a threshold $\tau > 0$ such that
\begin{equation} \label{sign_recovery}
\betahat_\tau^{\operatorname{lass0}(0,1)} \overset{\operatorname{s}}{=} \beta^0.
\end{equation}
\end{theorem}

The proof can be found in the appendix. Exact sign recovery~(\ref{sign_recovery}) clearly implies exact support recovery ($\Shat^{\rm{lass0}}_\tau = S^0$). Note that Theorem~\ref{main_thm} is purely deterministic, and deriving a consistency result requires to lower bound the probability that~(\ref{beta_min}) holds for a well chosen value of $\beta^0_{\min}.$ Ideally we would need a tail bound for $\norm{\betahatlone(\epsilon)}_1,$ however to the best of our knowledge such results are unknown. Instead, we derive Corollary~\ref{cor:corrGaussian} in Section~\ref{subsection:sign_consistency_Gauss} using $\norm{\betahatlone(\epsilon)}_1 \leq \norm{X^T(XX^T)^{-1} \epsilon}_1.$

We now compare the result in Theorem~\ref{main_thm} to the theory of the Lasso~(\ref{lasso}). It has been shown that in an asymptotic setting the $\theta$-irrepresentable condition
\begin{equation} \label{IR}
\norm{X_{\comp{S^0}}^T X_{S^0} \left(X_{S^0}^T X_{S^0} \right)^{-1} \sign(\beta^0_{S^0})}_\infty < \theta,
\end{equation}
where $\theta \in (0,1),$ implies consistent model selection by the Lasso \citep{zhao2006, zou2006}. Irrepresentability is actually almost necessary (the inequality in~(\ref{IR}) being replaced by $\leq 1$) even in the noiseless and finite sample case \citep{buhlmann2011}. Theorem~\ref{main_thm} holds no matter what $\sign(\beta^0_{S^0})$ is, whereas the irrepresentable condition~(\ref{IR}) is defined for specific signs. For the Lasso to consistently recover $S^0$ for arbitrary signs of $\beta^0_{S^0}$, the $\theta$-uniform irrepresentable condition
\begin{equation} \label{unifIR}
\max_{\norm{\tau_{S^0}}_\infty \leq 1} \norm{X_{\comp{S^0}}^T X_{S^0} \left(X_{S^0}^T X_{S^0} \right)^{-1} \tau_{S^0}}_\infty < \theta
\end{equation}
should be assumed. The following result confirms that the condition in Theorem~\ref{main_thm} is weaker than the $\theta$-uniform irrepresentable condition.

\begin{prop} \label{IR_SNSP}
The $\theta$-uniform irrepresentable condition~(\ref{unifIR}) implies the stable null space property~(\ref{SNSP}) for any $\rho \in (\theta, 1).$
\end{prop}
\begin{proof}
By Theorem 7.2. in \citet{buhlmann2011}, the $\theta$-uniform irrepresentable condition implies the so called $(L, S^0)$-\emph{compatibility condition} for any $L \in (1, \tfrac{1}{\theta})$, which states that there exists a $\phi^2(L, S^0) > 0$ such that
\begin{equation} \label{comp_cond}
\norm{\beta_{\comp{S^0}}}_1 \leq L \norm{\beta_{S^0}}_1
\Longrightarrow
\norm{\beta_{S^0}}_1^2 \leq \frac{\abs{S^0}}{n \phi^2(L, S^0)} \norm{X \beta}_2^2.
\end{equation}
So for any $\rho \in (\theta, 1),$ the compatibility condition holds for $L = \frac{1}{\rho}.$ Now assume there exists a $\beta \in \ker(X)$ for which $\norm{\beta_{S^0}}_1 > \rho \norm{\beta_{\comp{S^0}}}_1.$ Then we have $\norm{\beta_{\comp{S^0}}}_1 < L \norm{\beta_{S^0}}_1$ and the compatibility condition~(\ref{comp_cond}) implies that $\norm{\beta_{S^0}}_1^2 \leq \frac{\abs{S^0}}{n \phi^2(L, S^0)} \norm{X \beta}_2^2.$ But $\beta \in \ker(X)$ so we conclude $\beta_{S^0} = 0$, which contradicts the assumption $\norm{\beta_{S^0}}_1 > \rho \norm{\beta_{\comp{S^0}}}_1.$
\end{proof}

Proposition~\ref{IR_SNSP} shows that the stable null space property is not only weaker than the uniform irrepresentable condition, but also weaker than the compatibility condition, which is assumed for obtaining oracle inequality for the prediction and estimation error of the Lasso \citep{buhlmann2011}. See \citet{vandegeer2009} for more conditions used in the Lasso theory and the relations between them.

\subsection{Sign consistency for correlated Gaussian designs} \label{subsection:sign_consistency_Gauss}
We show here that Theorem~\ref{main_thm} implies sign consistency of Lasso-Zero for correlated Gaussian designs, in an asymptotic framework where both $p$ and $s^0 = \abs{S^0}$ grow with $n.$ More precisely, we make the following assumptions:
\begin{enumerate}[A)]
\item \label{assumption:Gaussdesign} the rows of $X \in \R^{n \times p}$ (with $n < p$) are random and i.i.d. $N_p(0, \Sigma),$
\item \label{assumption:boundedsigmamin} $\lambda_{\min}(\Sigma) \geq \gamma^2 > 0$ for every $p,$ 
\item \label{assumption:var1} $\Sigma_{ii}=1$ for every $i$,
\item \label{assumption:Gaussiannoise} $\epsilon \sim N_n(0, \sigma^2 I),$
\end{enumerate}
where in~\ref{assumption:boundedsigmamin}) we use $\lambda_{\min}$ to denote the smallest eigenvalue. Note that assumptions~\ref{assumption:Gaussdesign}) and \ref{assumption:boundedsigmamin}) imply that  almost surely $\rank{X} = n$. Assumption~\ref{assumption:boundedsigmamin}) typically holds for covariance matrices of the form $\Sigma_{ij} = a^{\abs{i-j}}$ or $\Sigma_{ij} = a + (1-a)\bm{1}_{\{i=j\}}$ with $a \in [0, 1).$ Assumption~\ref{assumption:var1}) is common and holds under an appropriate rescaling of the covariates.

\begin{corollary} \label{cor:corrGaussian}
Under model~(\ref{lin_model}) and assumptions~\ref{assumption:Gaussdesign}), \ref{assumption:boundedsigmamin}), \ref{assumption:var1}) and \ref{assumption:Gaussiannoise}), there exist universal constants $c, c', c'' > 0$ such that for any $\rho \in (0, 1)$
\begin{equation*}
\p\left(\exists \tau > 0 : \betahat^{\operatorname{lass0}(0, 1)}_\tau \overset{\operatorname{s}}{=} \beta^0 \right)
\geq 1 - c'e^{-cn} - 1.14^{-n} - 2e^{-\frac{1}{8}(\sqrt{p} - \sqrt{n})^2},
\end{equation*}
provided that
\begin{align}
n > c'' \frac{(1 + \rho^{-1})^2}{\gamma^2} s^0 \log{p},& \label{lb:n}
\\
\beta^0_{\min} > C_\rho \frac{2 \sqrt{2} \sigma \sqrt{p}}{\gamma (\sqrt{p/n} - 1)},& \label{beta_min_cor}
\end{align}
where $C_{\rho} = \frac{2(3 + \rho)}{1 - \rho}.$ Consequently, if $p/n \to C > 1,$ Lasso-Zero achieves sign consistency with $\beta^0_{\min} = \Omega(\sqrt{n})$ and $s^0 = O(n/\log{n})$\footnote{Recall that $f(n) = \Omega(g(n))$ means that there is $K > 0$ such that $\abs{f(n)} \geq K \abs{g(n)}$ for large enough $n$, and $f(n) = O(g(n))$ means that there is $k > 0$ such that $\abs{f(n)} \leq k \abs{g(n)}$ for large enough $n.$}.
\end{corollary}

Corollary~\ref{cor:corrGaussian} follows from Theorem~\ref{main_thm} and some results about Gaussian matrices that are reported in the appendix. The beta-min condition (\ref{beta_min_cor}) is strong and introducing noise dictionaries   improves thresholded BP for lower signals, unlike \citet{saligrama2011} who rather consider a multistage procedure where an ordinary least squares estimate is computed on the support of thresholded BP and is then thresholded again.

\subsection{Guarantees given by QUT in the low-dimensional case} \label{lowdim_FDR}

We derive here guarantees offered by QUT when $\rank{X} = p.$ By definition, QUT controls the FWER~(\ref{defFWER}) -- and therefore the FDR~(\ref{defFDR}) -- under the null model $\beta^0 = 0$. It turns out that this property extends to the low-dimensional case for any regression vector $\beta^0,$ provided all columns of $X$ are linearly independent.

\begin{prop} \label{control_FWER}
Consider the linear model~(\ref{lin_model}) with design matrix $X$ such that $\rank{X} = p.$ Then for any value of $\beta^0 \in \R^p$, Lasso-Zero tuned by QUT as in~(\ref{lass0_QUT}) and $(q, M) = (0, 1)$ controls the FWER, and therefore the FDR, at level $\alpha.$
\end{prop}
\begin{proof}
By Definition~\ref{deflass0} with $q=0$ and $M=1,$ Lasso-Zero thresholds the least-squares estimate $\betahat^{LS} = (X^T X)^{-1} X^T y = \beta^0 + (X^T X)^{-1} X^T \epsilon.$
Then for $\tau = \tau_\alpha^{\operatorname{QUT}}$ as in~(\ref{lass0_QUT}):
\begin{align*}
\operatorname{FWER} &= \p(\norm{\betahat^{LS}_{\comp{S^0}}}_\infty > \tau_\alpha^{\operatorname{QUT}}) 
  = \p(\norm{[(X^T X)^{-1} X^T \epsilon]_{\comp{S^0}}}_\infty > \tau_\alpha^{\operatorname{QUT}})
  \\
 & \leq \p(\norm{(X^T X)^{-1} X^T \epsilon}_\infty > \tau_\alpha^{\operatorname{QUT}})
 = \alpha.
\end{align*}
\end{proof}

The proof shows that Proposition~\ref{control_FWER} holds for any estimator that  thresholds the least-squares solution, e.g., for the Lasso when $X$ has orthogonal columns \citep{tibshirani1996}.

\section{Discussion} \label{discussion}
The main novelty provided by Lasso-Zero is to repeatedly use noise dictionaries and aggregate the obtained coefficients, here by taking the median componentwise  to detect the important predictors. This work provides promising results for the linear model, but noise dictionaries could be used in other settings, such as generalized linear models, and for pruning trees in random forest or introducing regularization in artificial neural networks.

Because Lasso-Zero is based on the limit of the Lasso path at zero, the shrinkage effect induced by the $\ell_1$-penalty is reduced. In a discussion about our simulations, Rob Tibshirani raised our attention to an extensive comparison of Lasso and best subset selection \citep{hastie2017}.
Their results suggest that best subset (least bias) performs better in high $\operatorname{SNR}$, whereas the Lasso outperforms it when the $\operatorname{SNR}$ is low. 
So one might expect Lasso to outperform Lasso-Zero as well for low $\operatorname{SNR}$. Nonetheless, preliminary results seem to indicate that Lasso-Zero can be improved by tuning the noise dictionaries' size $q$ in some adaptive way. 
Increasing or decreasing it has an influence on the shrinkage of the estimated coefficients and therefore leaves a room for improvement in both low and high $\operatorname{SNR}$ regimes.

\section{Reproducible research}
The R package \texttt{lass0} (also on CRAN) and the code that generated the figures in this article may be found at \url{https://github.com/pascalinedescloux/lasso-zero}.

\appendix

\section{Proof of Theorem~\ref{main_thm}}

\begin{lemma} \label{aux_lemma}
Under the assumptions of Theorem~\ref{main_thm},
\[
\norm{\betahatlone(y) - (\beta^0 + \betahatlone(\epsilon))}_1 
\leq 
\frac{2(1+\rho)}{1-\rho} \norm{\betahatlone(\epsilon)}_1.
\]
\end{lemma}
\begin{proof}
For any $\beta \in \R^p$, we have
\begin{align*}
\norm{\betahatlone(y)_{\comp{S^0}}}_1 
&=
\norm{(\betahatlone(y) - \beta)_{\comp{S^0}} +\beta_{\comp{S^0}}}_1 \\
& \geq
\norm{(\betahatlone(y) - \beta)_{\comp{S^0}}}_1 - \norm{\beta_{\comp{S^0}}}_1
\end{align*}
If $\beta$ satisfies $y = X \beta$, then
\begin{align*}
\norm{\betahatlone(y)_{S^0}}_1 + \norm{(\betahatlone(y) - \beta)_{\comp{S^0}}}_1 - \norm{\beta_{\comp{S^0}}}_1
& \leq
\norm{\betahatlone(y)_{S^0}}_1 + \norm{\betahatlone(y)_{\comp{S^0}}}_1
 =
\norm{\betahatlone(y)}_1 \\ 
& \leq 
\norm{\beta}_1 
 = 
\norm{\beta_{S^0}}_1 + \norm{\beta_{\comp{S^0}}}_1,
\end{align*}
where the last inequality holds since $\beta$ is feasible for BP~(\ref{basis_pursuit}). Rearranging terms yields
\begin{align*}
\norm{(\betahatlone(y) - \beta)_{\comp{S^0}}}_1
& \leq
\norm{\beta_{S^0}}_1 - \norm{\betahatlone(y)_{S^0}}_1 + 2 \norm{\beta_{\comp{S^0}}}_1 \\
& \leq
\norm{(\betahatlone(y) - \beta)_{S^0}}_1 + 2 \norm{\beta_{\comp{S^0}}}_1.
\end{align*}
Using the feasible vector $\beta = \beta^0 + \betahatlone(\epsilon)$, we obtain
\[
\norm{(\betahatlone(y) - \betahatlone(\epsilon))_{\comp{S^0}}}_1
\leq 
\norm{\betahatlone(y)_{S^0} - (\beta^0 + \betahatlone(\epsilon))_{S^0}}_1 + 2 \norm{\betahatlone(\epsilon)_{\comp{S^0}}}_1.
\]
Now $\betahatlone(y) - (\beta^0 + \betahatlone(\epsilon))$ belongs to $\ker(X)$, so by the stable null space property~(\ref{SNSP}) $\norm{\betahatlone(y)_{S^0} - (\beta^0 + \betahatlone(\epsilon))_{S^0}}_1 \leq \rho \norm{(\betahatlone(y) - \betahatlone(\epsilon))_{\comp{S^0}}}_1$, hence
\begin{equation} \label{aux_bound}
\norm{(\betahatlone(y) - \betahatlone(\epsilon))_{\comp{S^0}}}_1
\leq 
\frac{2}{1-\rho} \norm{\betahatlone(\epsilon)_{\comp{S^0}}}_1
\leq \frac{2}{1-\rho} \norm{\betahatlone(\epsilon)}_1.
\end{equation}
The stable null space property also implies
\begin{align*}
\norm{\betahatlone(y) - (\beta^0 + \betahatlone(\epsilon))}_1
& = 
\norm{(\betahatlone(y) - (\beta^0 + \betahatlone(\epsilon)))_{S^0}}_1
+\norm{(\betahatlone(y) - \betahatlone(\epsilon))_{\comp{S^0}}}_1 \\
& \leq 
(1+\rho) \norm{(\betahatlone(y) - \betahatlone(\epsilon))_{\comp{S^0}}}_1.
\end{align*}
Combining this with the inequality~(\ref{aux_bound}) gives the statement.
\end{proof}

\begin{proof}[Proof of Theorem~\ref{main_thm}]
For every $j \in \{1, \ldots, p\}$ we have
\begin{equation*}
\begin{aligned}
\abs{\betahatlone_j(y) - \beta^0_j} 
&= 
\abs{\betahatlone_j(y) - (\beta^0_j + \betahatlone_j(\epsilon)) + \betahatlone_j(\epsilon)} 
\\
&\leq 
\abs{\betahatlone_j(y) - (\beta^0_j + \betahatlone_j(\epsilon))} + \abs{\betahatlone_j(\epsilon)}
\\
& \leq 
\norm{\betahatlone(y) - (\beta^0 + \betahatlone(\epsilon))}_1 + \norm{\betahatlone(\epsilon)}_1
\\
& \leq 
\frac{2(1+\rho)}{1-\rho} \norm{\betahatlone(\epsilon)}_1 +  \norm{\betahatlone(\epsilon)}_1
\\
&= 
\frac{3 + \rho}{1-\rho} \norm{\betahatlone(\epsilon)}_1,
\end{aligned}
\end{equation*}
where we have used Lemma~\ref{aux_lemma} in the last inequality. Choosing $\tau := \frac{3+\rho}{1-\rho} \norm{\betahatlone(\epsilon)}_1,$ this means
\begin{equation} \label{ineq:diff_coef}
\abs{\betahatlone_j(y) - \beta^0_j}  \leq \tau.
\end{equation}
So for $j \in \comp{S^0},$ $\abs{\betahatlone_j(y)} \leq \tau$ and therefore $\betahat^{\textrm{lass0}}_{\tau, j}(y) = 0.$ Now for $j \in S^0,$ assumption~(\ref{beta_min}) gives $\abs{\beta^0_j} > 2 \tau.$ Combining this with~(\ref{ineq:diff_coef}) implies that $\sign(\betahatlone_j(y)) = \sign(\beta^0_j)$ and $\abs{\betahatlone_j(y)} > \tau,$ hence $\sign(\betahat^{\textrm{lass0}}_{\tau, j} (y)) = \sign(\beta^0_j).$

\end{proof}

\section{Proof of Corollary~\ref{cor:corrGaussian}}

The stable null space property holds with high probability by the following result.
\begin{lemma} \label{lemma:NSP}
Assuming~\ref{assumption:Gaussdesign}), \ref{assumption:boundedsigmamin}) and \ref{assumption:var1}), there exist universal constants $c, c', c'' > 0$ such that for any $\rho \in (0, 1),$ $X$ satisfies the stable null space property with respect to $S^0$ and with constant $\rho$ with probability at least $1- c'e^{-cn},$ provided~(\ref{lb:n}) holds.
\end{lemma}
\begin{proof}
It follows directly from Corollary 1 in \citet{raskutti2010}, which proves that under our assumptions $\frac{1}{n}X^T X$ satisfies the restricted eigenvalue condition with parameters $(\rho^{-1}, \gamma/8)$ with high probability, which implies the stable null space property for $X$.
\end{proof}

Next result is useful to prove the beta-min condition~(\ref{beta_min}) holds with high probability.
\begin{lemma} \label{lemma:sing_val_Gauss}
Let $G$ be an $n \times p$ matrix with $n < p$ and i.i.d.~entries $N(0, 1).$ Then the smallest singular value $\sigma_{\min}(G)$ of $G$ satisfies
\[
\p\left( \sigma_{\min}(G) \geq \frac{1}{2} (\sqrt{p} - \sqrt{n})\right) \geq 1 - 2e^{-\frac{1}{8}(\sqrt{p} - \sqrt{n})^2}.
\]
\end{lemma}
\begin{proof}
See \citet[eq. (2.3)]{rudelson2010}.
\end{proof}

The proof of Corollary~\ref{cor:corrGaussian} is now a straightforward adaptation of the one of Lemma V.4 in \citet{saligrama2011} but is included for the sake of completeness.
\begin{proof}[Proof of Corollary~\ref{cor:corrGaussian}]
By~\ref{assumption:Gaussdesign}), one can write $X =  G \Sigma^{1/2}$ with $G$ as in Lemma~\ref{lemma:sing_val_Gauss}, and using assumption~\ref{assumption:boundedsigmamin}) we get
$
\sigma_{\min}(X) \geq  \sigma_{\min}(\Sigma^{1/2}) \sigma_{\min}(G) \geq \gamma \sigma_{\min}(G).
$
So by Lemma~\ref{lemma:sing_val_Gauss}
\begin{equation}\label{prob:bound_sigmamin}
\p\left(X \in \mathcal{X} \right) \geq 1- 2e^{-\frac{1}{8}(\sqrt{p} - \sqrt{n})^2},
\end{equation}
where
$\mathcal{X} := \{X \in \R^{n \times p} \ \mid \ \sigma_{\min}(X) \geq \frac{\gamma \sqrt{n}}{2} (\sqrt{p/n} - 1)  \}.$ 

For now consider a fixed $X \in \mathcal{X}.$ The vector $\omega := X^T (X X^T)^{-1}\epsilon$ being feasible for~(\ref{opt:BP_epsilon}), 
\begin{equation}\label{ineq:norm1_betaBP}
\norm{\betahatlone(\epsilon)}_1 \leq \norm{\omega}_1.
\end{equation}
We write the SVD decomposition $X^T = U \Lambda V^T,$ with $U, V$ square orthonormal matrices of size $p$ and $n$ respectively, $\Lambda = \left(\begin{smallmatrix} D \\ 0_{(p-n)\times n} \end{smallmatrix}\right)$ with $D = \diag(\sigma_1(X), \ldots, \sigma_n(X))$ 
and where $0_{(p-n)\times n}$ is the zero matrix of size $(p-n) \times n.$ Then 
$
\omega = U  \Lambda^{-1} V^T \epsilon,
$
where with a slight abuse of notation 
$
\Lambda^{-1} = 
\left( \begin{smallmatrix} D^{-1} \\ 0_{(p-n)\times n} \end{smallmatrix} \right).
$
Now $U$ is orthonormal, so $\norm{\omega}_2 = \norm{\Lambda^{-1} V^T \epsilon}_2 = \norm{D^{-1} V^T \epsilon}_2.$ By assumption~\ref{assumption:Gaussiannoise}), the vector $\tilde{\omega} := D^{-1} V^T \epsilon$ follows a multivariate Gaussian distribution with covariance matrix given by $\diag(\frac{\sigma^2}{\sigma_1^2(X)}, \ldots, \frac{\sigma^2}{\sigma_n^2(X)}).$ We have
\begin{equation*}
\begin{aligned}
\norm{\omega}_2^2 = \norm{\tilde{\omega}}_2^2 = \sum_{i=1}^n \tilde{\omega}_i^2 &= \sum_{i=1}^n\frac{\sigma^2}{\sigma_i^2(X)} \left(\frac{\sigma_i(X)}{\sigma} \tilde{\omega}_i\right)^2 
\leq \frac{4 \sigma^2}{n \gamma^2 (\sqrt{p/n} - 1)^2} \sum_{i=1}^n\left(\frac{\sigma_i(X)}{\sigma} \tilde{\omega}_i\right)^2,
\end{aligned}
\end{equation*}
where the inequality holds since $X \in \mathcal{X}.$ Now $\sum_{i=1}^n\left(\frac{\sigma_i(X)}{\sigma} \tilde{\omega}_i\right)^2 \sim \chi^2_n$ so it is upper bounded by $2n$ with probability larger than $1 - 1.14^{-n}$ (this is a corollary of Lemma 1 in \citet{laurent2000}), and therefore
\begin{equation} \label{bound:norm2_omega}
\p\left(\norm{\omega}_2 > \frac{2\sqrt{2} \sigma}{\gamma (\sqrt{p/n} -1)}\right) \leq 1.14^{-n}.
\end{equation}
Recall that~(\ref{bound:norm2_omega}) holds for fixed $X \in \mathcal{X}.$ Now for random $X$ we have
\begin{equation}
\begin{aligned}
\p\left( \norm{\omega}_1 > \frac{2 \sqrt{2} \sigma \sqrt{p}}{\gamma (\sqrt{p/n} - 1)} \right)
&\leq
\p \left(\sqrt{p} \norm{\omega}_2 >  \frac{2 \sqrt{2} \sigma \sqrt{p}}{\gamma (\sqrt{p/n} - 1)}\right)
\\
&= 
\p \left(\norm{\omega}_2 >  \frac{2 \sqrt{2} \sigma}{\gamma (\sqrt{p/n} - 1)}\right)
\\
&\leq
\p \left(\norm{\omega}_2 >  \frac{2 \sqrt{2} \sigma}{\gamma (\sqrt{p/n} - 1)} \ \mid \ X \in \mathcal{X}\right) + \p(X \notin \mathcal{X}).
\\
& \leq
1.14^{-n} + 2e^{-\frac{1}{8}(\sqrt{p} - \sqrt{n})^2}, \label{prob:norm1_omega}
\end{aligned}
\end{equation}
where the last inequality comes from~(\ref{prob:bound_sigmamin}) and the fact that~(\ref{bound:norm2_omega}) holds for every $X \in \mathcal{X}.$ From~(\ref{beta_min_cor}), (\ref{ineq:norm1_betaBP}) and~(\ref{prob:norm1_omega}), we get
$
\p\left( \beta^0_{\min} \leq C_{\rho} \norm{\betahatlone(\epsilon)}_1\right) \leq 1.14^{-n} + 2e^{-\frac{1}{8}(\sqrt{p} - \sqrt{n})^2}$.
Together with Lemma~\ref{lemma:NSP}, this guarantees that all conditions of Theorem~\ref{main_thm} hold with probability at least $1 - c'e^{-cn} - 1.14^{-n} - 2e^{-\frac{1}{8}(\sqrt{p} - \sqrt{n})^2},$ which concludes our proof.
\end{proof}

\section*{Acknowledgements}
Research funded by the Department of Mathematics of the University of Geneva. Computations were performed at University of Geneva on the Baobab cluster. We thank the Statistics Department of Stanford University for receiving the first author as a visitor; Rob Tibshirani, Trevor Hastie, Jonathan Taylor, as well as Claire Boyer and Julie Josse for interesting discussions.  

\bibliographystyle{chicago}
\bibliography{lass0}
\end{document}